\documentclass[amstex,12pt]{article}
\usepackage{amssymb}
\usepackage{amsmath}
\usepackage{amsthm}

\usepackage[normalem]{ulem}

\numberwithin{equation}{section}

\newcommand{\be}{\begin{equation}}
\newcommand{\ee}{\end{equation}}
\newcommand{\bea}{\begin{eqnarray}}
\newcommand{\eea}{\end{eqnarray}}
\newcommand{\non}{\nonumber}
\newcommand{\id}{\mathbb{I}}
\newcommand{\C}{\mathbb{C}}
\newcommand{\tr}{\mathop{\rm tr}\nolimits}
\newcommand{\diag}{\mathop{\rm diag}\nolimits}

\newcommand{\half}{\frac{1}{2}}

\usepackage[usenames,dvipsnames]{color}

\newtheorem{lemma}{Lemma}
\newtheorem*{prop*}{Proposition}
\newtheorem{prop}{Proposition}
\newtheorem*{corollary*}{Corollary}

\newcommand\blfootnote[1]{%
  \begingroup
  \renewcommand\thefootnote{}\footnote{#1}%
  \addtocounter{footnote}{-1}%
  \endgroup
}

\begin{document}

\begin{titlepage}
\strut\hfill UMTG--294
\vspace{.5in}
\begin{center}

\LARGE Surveying the quantum group symmetries\\
of integrable open spin chains\\
\vspace{1in}
\large 
Rafael I. Nepomechie \footnote{
Physics Department,
P.O. Box 248046, University of Miami, Coral Gables, FL 33124 USA}
and Ana L. Retore ${}^{1,}$\footnote{
Instituto de F\'{i}sica Te\'{o}rica-UNESP, Rua 
Dr. Bento Teobaldo Ferraz 271, Bloco II 01140-070, S\~{a}o Paulo, Brazil}\\[0.8in]
\end{center}

\vspace{.5in}

\begin{abstract}
Using anisotropic R-matrices associated with affine Lie algebras $\hat g$
(specifically, $A_{2n}^{(2)}$, $A_{2n-1}^{(2)}$, $B_{n}^{(1)}$,
$C_{n}^{(1)}$, $D_{n}^{(1)}$) and suitable 
corresponding K-matrices, we construct families of integrable
open quantum spin chains of finite length, whose transfer matrices are invariant under
the quantum group corresponding to removing one node from the Dynkin
diagram of $\hat g$.  We show that these transfer matrices also have a
duality symmetry (for the cases $C_{n}^{(1)}$ and $D_{n}^{(1)}$) and
additional $Z_{2}$ symmetries that map complex representations to
their conjugates (for the cases $A_{2n-1}^{(2)}$, $B_{n}^{(1)}$ and
$D_{n}^{(1)}$).  A key simplification is achieved by working in a
certain ``unitary'' gauge, in which only the unbroken symmetry
generators appear. The proofs of these symmetries rely on some new properties of the R-matrices.
We use these symmetries to explain the
degeneracies of the transfer matrices. 
\end{abstract}

\blfootnote{e-mail addresses: {\tt nepomechie@miami.edu, retore@ift.unesp.br}}

\end{titlepage}

\setcounter{footnote}{0}

\section{Introduction and summary}\label{sec:intro}

Quantum spin chains have numerous applications.  Being interacting
many-body systems, their spectra are generally difficult to determine
when the number of spins is large.  The simplest anisotropic spin
chains are arguably those that are integrable and have quantum group
(QG) symmetries.  Indeed, integrability can help to determine the spectrum,
and QG symmetry can help to explain the degeneracies and
multiplicities.  The first such example was the
$U_{q}(A_{1})$-invariant open spin-1/2 chain \cite{Pasquier:1989kd,
Kulish:1991np}, whose integrability follows from \cite{Alcaraz:1987uk,
Sklyanin:1988yz}.  Various higher-rank generalizations have been
investigated, see e.g. \cite{Mezincescu:1990ui, Mezincescu:1991rb,
Mezincescu:1991ag, Foerster:1993fp, deVega:1993ae, devega:1994hf, 
GonzalezRuiz:1994tw,
Artz:1994qy, Artz:1995bm, Yung:1994tm, Yung:1995ag, Batchelor:1995gx, 
Doikou:1998ek, Nepomechie:1999jz, Martins:2000xie, Li:2003ucc, 
Li:2004jsb, Kurak:2004ip, Doikou:2005pn,
Li:2005pp, Li:2006mv, Ahmed:2017mqq,
Nepomechie:2017hgw}.\footnote{All of these spin chains are 
open, since closed anisotropic integrable spin chains of finite
length with periodic boundary conditions generally do not have such QG
symmetry.}

We identify here new families of integrable QG-invariant 
spin chains, which include as special cases many of the previously-studied models.
Specifically, we construct integrable open spin chains of finite length
using anisotropic R-matrices associated with several families of affine Lie 
algebras $\hat g$
(namely, $A_{2n}^{(2)}$, $A_{2n-1}^{(2)}$, $B_{n}^{(1)}$, $C_{n}^{(1)}$, 
$D_{n}^{(1)}$) \cite{Jimbo:1985ua, Bazhanov:1986mu, Kuniba:1991yd, 
Costello:2017dso, Costello:2018gyb},
and with corresponding diagonal K-matrices depending on an 
integer $p \in [0, n]$ \cite{Batchelor:1996np, LimaSantos:2002ui, LimaSantos:2003hx,
Malara:2004bi}, whose transfer matrices $t(u,p)$ (\ref{transfer}) are
invariant under the QG corresponding to removing the $p^{th}$ node from the (extended) 
Dynkin diagram of $\hat g$, as summarized in the middle column of Table \ref{table:symmetries}.

\begin{table}[h]
\centering
\begin{tabular}{|c|c|c|}
\hline
$\hat g$ & QG symmetry & Representation at each site\\   
\hline
$A_{2n}^{(2)}$ & $U_{q}(B_{n-p}) \otimes U_{q}(C_{p}) $ & $(2(n-p)+1,1) \oplus (1,2p)$ \\
$A_{2n-1}^{(2)}$ & $U_{q}(C_{n-p}) \otimes U_{q}(D_{p})\quad (p \ne 1)$ 
& $(2(n-p),1) \oplus (1,2p)$ \\
$B_{n}^{(1)}$ & $U_{q}(B_{n-p}) \otimes U_{q}(D_{p}) \quad (n>1, p 
\ne 1)$ & $(2(n-p)+1,1) \oplus (1,2p)$ \\
$C_{n}^{(1)}$ & $U_{q}(C_{n-p}) \otimes U_{q}(C_{p}) $ & $(2(n-p),1) \oplus (1,2p)$ \\
$D_{n}^{(1)}$ & $U_{q}(D_{n-p}) \otimes U_{q}(D_{p}) \quad (n>1, p \ne 1, 
n-1)$ & $(2(n-p),1) \oplus (1,2p)$ \\
\hline
\end{tabular}
\caption{\small QG symmetries of the open-chain transfer matrix, where $p=0, 1, 
\ldots, n$.}\label{table:symmetries}
\end{table}

For $p=0$, the ``right'' (second) factors in Table \ref{table:symmetries} are absent; these cases
were studied long ago \cite{Mezincescu:1990ui, Mezincescu:1991rb, Mezincescu:1991ag,
Artz:1994qy, Artz:1995bm}. For $p=n$, the ``left'' (first) factors in Table \ref{table:symmetries} are absent;
these cases were noticed only recently \cite{Nepomechie:1999jz, Ahmed:2017mqq, 
Nepomechie:2017hgw}. For intermediate values $0 < p < n$, the QG
symmetries are generally given by a tensor product of two factors, as 
shown in Table \ref{table:symmetries}; these cases 
had not been considered until now. Moreover, we prove for all $p \in 
[0, n]$ that the transfer matrices have these QG 
symmetries.\footnote{Such proofs had 
been known, following \cite{Kulish:1991np}, only for the cases with 
$p=0$ \cite{Mezincescu:1991rb}. For the 
cases with $p=n$, the QG symmetry was conjectured for the transfer 
matrices, but was proved only for the corresponding Hamiltonians \cite{Ahmed:2017mqq, 
Nepomechie:2017hgw}.} 

A key role in our proof of the QG symmetry of the transfer matrix is
played by so-called gauge transformations of the R-matrix and
K-matrices.  By transforming to a certain ``unitary'' gauge, the asymptotic monodromy matrix
becomes expressed in terms of only the unbroken symmetry generators, which
then allows us to invoke the powerful machinery of the Quantum 
Inverse Scattering Method (QISM).
The relevance of this gauge transformation can already be seen from the following
observation: the R-matrices of Jimbo \cite{Jimbo:1985ua}, which are in
the so-called homogeneous ``picture'' or ``gradation'' (gauge), have
the symmetries in Table \ref{table:symmetries} with $p=0$ (more
precisely, $\check{R} = {\cal P} R$ commutes with the coproducts of
the generators of the QG); while the gauge-transformed
R-matrices, with the gauge transformation corresponding to $p=n$ (see
Eqs. (\ref{gaugeR}) and (\ref{Bgauge}) below), have instead the symmetries in Table
\ref{table:symmetries} with $p=n$.

For the cases $C_{n}^{(1)}$ and $D_{n}^{(1)}$, the two symmetry
factors in Table \ref{table:symmetries} evidently interchange under $p
\leftrightarrow n-p$.  In fact, we show that the corresponding
transfer matrices are related by ``duality'' transformations 
(\ref{duality}), implying
that their spectra are equal.  For the special case with $n$ even and 
$p=\frac{n}{2}$, the transfer matrix is self-dual (\ref{selfdual}), which gives rise to
degeneracies in the spectrum beyond those expected from QG symmetry.

For the cases that at least one of the symmetry factors in Table
\ref{table:symmetries} is of type $D$ (namely, $A_{2n-1}^{(2)}$,
$B_{n}^{(1)}$ and $D_{n}^{(1)}$), we show that the transfer matrices
have additional $Z_{2}$ symmetries that map complex representations to
their conjugates (\ref{Z2rightprop}), (\ref{Z2leftprop}). 

All of these symmetries are useful for understanding the degeneracies
in the spectrum of the transfer matrix.  In proving these symmetries,
we use various properties of the R-matrices (\ref{RmatpropCD}), (\ref{Z2rightR}),
(\ref{Z2leftR}), which to our knowledge are new, and which may be 
of independent interest.

The outline of this paper is as follows.  The transfer matrix is
introduced in Sec.  \ref{sec:basics}.  The QG symmetry of the transfer
matrix is proved in Sec.  \ref{sec:proof}.  The duality symmetry of
the transfer matrix (for the cases $C_{n}^{(1)}$ and $D_{n}^{(1)}$), 
and the action of duality on the QG generators,
are worked out in Sec.  \ref{sec:duality}. The additional $Z_{2}$ 
symmetries of the transfer matrix (for the cases $A_{2n-1}^{(2)}$,
$B_{n}^{(1)}$ and $D_{n}^{(1)}$), and the action of these symmetries 
on the QG generators, are worked out in Sec.  \ref{sec:Z2}. 
These symmetries are used in  Sec. \ref{sec:degen} to explain the 
degeneracies in the spectrum of 
the transfer matrix for generic values of the anisotropy parameter 
$\eta$. Some interesting remaining open problems are listed in Sec. 
\ref{sec:outlook}. The R-matrices are recalled in Appendix 
\ref{sec:Rmatices}, details about the QG generators are presented 
in Appendix \ref{sec:QG}, and the Hamiltonian is noted 
in Appendix \ref{sec:Hamiltonian}. Proofs of several lemmas are 
outlined in Appendix \ref{sec:proofs}. 

\section{Basics}\label{sec:basics}

We consider an integrable open quantum spin chain with a vector space ${\cal 
V} = \C^{d}$ at each of its $N$ sites, where
\begin{align}
d = \left\{ \begin{array}{ll}
2n+1 & \mbox{ for  } A_{2n}^{(2)}\,, B_{n}^{(1)} \\
2n & \mbox{ for  } A_{2n-1}^{(2)}\,, C_{n}^{(1)} \,, D_{n}^{(1)}  
\end{array} \right. \,, \qquad n = 1, 2, \ldots
\,.
\label{defd}
\end{align}
The Hilbert space (``quantum'' space) of the spin chain is therefore ${\cal V}^{\otimes 
N}$.

\subsection{R-matrix}

The bulk interactions of the spin chain are encoded in the
R-matrix $R(u)$, which maps ${\cal V} \otimes {\cal V}$ to itself,
and satisfies the Yang-Baxter equation (YBE) on ${\cal V} \otimes {\cal V}\otimes {\cal V}$
\be
R_{12}(u - v)\,  R_{13}(u)\, R_{23}(v) = R_{23}(v)\,  R_{13}(u)\, R_{12}(u - v)
\,.  \label{YBE}
\ee
We use the standard notations $R_{12} = R \otimes \id\,, R_{23} = \id \otimes R\,, R_{13} = 
{\cal P}_{23}  R_{12} {\cal P}_{23} = {\cal P}_{12}  R_{23} {\cal 
P}_{12} $, where $\id$ is the identity 
matrix on ${\cal V}$, and ${\cal P}$ is the permutation matrix on ${\cal V} \otimes {\cal V}$
\be
{\cal P}=\sum_{i, j = 1}^d e_{i j}\otimes e_{j i} \,,
\label{permutation}
\ee
where $e_{i j}$ are the $d \times d$ elementary 
matrices with elements $(e_{i j})_{\alpha \beta} = 
\delta_{i, \alpha} \delta_{j, \beta}$. 

We consider here the anisotropic R-matrices (with anisotropy parameter 
$\eta$) corresponding to the following affine Lie algebras 
\footnote{We do not consider here the case $A_{n}^{(1)}$, which does 
not have crossing symmetry; it has been studied in a 
similar context in \cite{deVega:1993ae, devega:1994hf, Doikou:1998ek}.}
\be
\hat g = \{ A_{2n}^{(2)}\,,  A_{2n-1}^{(2)}\,,  B_{n}^{(1)}\,, C_{n}^{(1)}\,, 
D_{n}^{(1)} \}\,.
\label{affinealgebras}
\ee
These R-matrices, which are given by Jimbo \cite{Jimbo:1985ua} 
(except for $A_{2n-1}^{(2)}$, in which case we consider instead  
Kuniba's R-matrix \cite{Kuniba:1991yd}), are in the 
homogeneous picture (gauge).\footnote{Bazhanov's R-matrices  
\cite{Bazhanov:1986mu} are equivalent, but are instead in the principal picture.}
These R-matrices, which can be found in Appendix \ref{sec:Rmatices},
all have the following additional properties: $PT$ symmetry
\be
R_{21}(u) \equiv {\cal P}_{12}\, R_{12}(u)\, {\cal P}_{12} 
= R_{12}^{t_1 t_2}(u) \,,
\label{PT}
\ee
unitarity
\be
R_{12}(u)\ R_{21}(-u) = \zeta(u)\, \id\otimes\id  \,,
\label{unitarity}
\ee
where $\zeta(u)$ is given by
\be
\zeta(u) =\xi(u)\, \xi(-u)\,, \qquad 
\xi(u)=
-2 \,\delta_1\,\sinh(\frac{1}{2}(u +4\eta)) \sinh(\frac{1}{2}(u  +\rho)) \,,
\ee
where $\delta_1$ is given by
\be
\delta_1=\begin{cases}
	i  &  \mbox{ for  } A_{2n}^{(2)}\,, A_{2n-1}^{(2)} \\
	1 &  \mbox{ for  } B_{n}^{(1)}\,, C_{n}^{(1)}\,, D_{n}^{(1)} 
	\,,
\end{cases}
\ee 
and crossing symmetry
\be
R_{12}(u)=V_1\, R_{12}^{t_2}(-u-\rho)\, V_1
= V_2^{t_2}\, R_{12}^{t_1}(-u-\rho)\, V_2^{t_2} \,,
\label{crossing}
\ee
where the crossing parameter $\rho$ is given by
\be
\rho= \left\{ \begin{array}{ll}
-2 \kappa \eta  - i \pi &  \mbox{ for  } A_{2n}^{(2)}\,, A_{2n-1}^{(2)} \\
-2 \kappa \eta & \mbox{ for  } B_{n}^{(1)}\,, C_{n}^{(1)}\,, D_{n}^{(1)}
\end{array} \right. \,,
\label{rho}
\ee
with $\kappa$ defined in (\ref{kappa}).
The crossing matrix $V$ is an antidiagonal matrix given by
\be
V= \delta_2 \sum_{\alpha=1}^{d} \epsilon_{\alpha} e^{(\bar \alpha - 
\bar{\alpha'}) \eta} 
e_{\alpha \alpha'} \,, \qquad V^2 = \id \,,
\ee
where $\delta_2$ is given by
\be
\delta_2 = \left\{ \begin{array}{ll}
1 & \mbox{ for } A_{2n}^{(2)}\,, B_{n}^{(1)}\,,  D_{n}^{(1)}\non \\
i & \mbox{ for } A_{2n-1}^{(2)}\,, C_{n}^{(1)}
\end{array}\right. \,,
\ee
and the other notations are defined in (\ref{epsilonalpha})-(\ref{prime}).
The corresponding matrix $M$ is defined by
\be
M = V^{t}\, V \,,
\label{Mdef}
\ee
and it is given by the diagonal matrix 
\be
M = \delta_2^2\sum_{\alpha=1}^{d} e^{4(\frac{d+1}{2}-\bar \alpha)\eta} e_{\alpha 
\alpha} \,.
\ee

\subsection{K-matrices}

The boundary interactions are encoded in the right and 
left  K-matrices, denoted here by $K^{R}(u)$ and $K^{L}(u)$, 
respectively, which map ${\cal V}$ to itself.\footnote{Following Sklyanin \cite{Sklyanin:1988yz}, 
the right and left K-matrices are usually denoted 
instead by $K^{-}(u)$ and $K^{+}(u)$, respectively. However, we 
adopt a different notation here in order to avoid 
confusion with the $\pm$ used in subsequent sections
to denote the limits $u \rightarrow \pm \infty$.} We choose $K^{R}(u)$ 
to be the diagonal $d \times d$ matrix
\be
K^{R}(u) = K^{R}(u,p) = \diag \big( \underbrace{e^{-u}\,, \ldots\,, 
e^{-u}}_{p}\,, \underbrace{\frac{\gamma e^{u} + 1}{\gamma  + 
e^{u}}\,, \ldots\,, \frac{\gamma e^{u} + 1}{\gamma  + 
e^{u}}}_{d-2p}\,, \underbrace{e^{u}\,, \ldots\,, 
e^{u}}_{p}\big) \,,
\label{KR}
\ee
where $p = 0, 1, \ldots, n$, and 
\begin{align}
\gamma = \left\{ \begin{array}{cl}
\gamma_{0} e^{(4p-2)\eta + \frac{1}{2}\rho}& \mbox{ for  } 
A_{2n-1}^{(2)}\,, B_{n}^{(1)}\,, D_{n}^{(1)}  \\ \\
\gamma_{0} e^{(4p+2)\eta + \frac{1}{2}\rho} & \mbox{ for  } A_{2n}^{(2)}\,, C_{n}^{(1)}  
\end{array} \right. \,, \qquad \gamma_{0} = \pm 1\,,
\label{gamma}
\end{align}
where $\rho$ is the crossing parameter (\ref{rho}). Unless otherwise 
noted, all the results in this paper hold for both values ($\pm 1$) of the  
parameter $\gamma_{0}$. As observed in \cite{Malara:2004bi} (see also \cite{Batchelor:1996np, LimaSantos:2002ui, LimaSantos:2003hx}), 
the matrices (\ref{KR}) are solutions of the boundary Yang-Baxter 
equation (BYBE) on ${\cal V} \otimes {\cal V}$
\cite{Sklyanin:1988yz, Cherednik:1985vs, Ghoshal:1993tm}
\be
R_{12}(u - v)\, K^{R}_1(u)\ R_{21} (u + v)\, K^{R}_2(v)
= K^{R}_2(v)\, R_{12}(u + v)\, K^{R}_1(u)\, R_{21}(u - v)  \,.
\label{BYBEm}
\ee
For $p=0$, we see that $K^{R}(u,p)$ in (\ref{KR}) is proportional to the identity 
matrix,
\be
K^{R}(u,0) \propto \id \,,
\label{KRpzero}
\ee
which is the solution noted in \cite{Mezincescu:1990ui}. We emphasize 
that the solution (\ref{KR}) depends on the bulk anisotropy parameter 
$\eta$ and the discrete boundary parameters $p$ and 
$\gamma_{0}$, but does not have any 
continuous boundary parameters.

For the left K-matrix, we take
\be
K^{L}(u) = K^{L}(u,p) = K^{R}(-u-\rho,p)\, M\,,
\label{KL}
\ee
where $M$ is given by (\ref{Mdef}), which is a solution of the 
corresponding BYBE \cite{Sklyanin:1988yz, Mezincescu:1990uf}
\begin{align}
\lefteqn{R_{12}(-u + v)\, K_1^{L\, t_1}(u)\, M^{-1}_1\, R_{21} (-u -v 
-2\rho)\, M_1\, K_2^{L\, t_2}(v)} \non \\
&  \quad = K^{L\, t_2}_2(v)\, M_1\, R_{12}(-u - v- 2\rho)\, M^{-1}_1\,
K^{L\, t_1}_1(u)\, R_{21}(-u +v)  \,.
\label{BYBEp}
\end{align} 

\subsection{Transfer matrix}

The open-chain transfer matrix, which maps the quantum space 
${\cal V}^{\otimes N}$ to itself, is given by\cite{Sklyanin:1988yz}
\be
t(u,p) = \tr_a K^{L}_{a}(u,p)\, T_a(u)\,  K^{R}_{a}(u,p)\, \widehat{T}_a(u) \,, 
\label{transfer}
\ee
where the single-row monodromy matrices are defined by
\begin{align} 
T_a(u) &= R_{aN}(u)\ R_{a N-1}(u)\ \cdots R_{a1}(u) \,,  \non \\
\widehat{T}_a(u) &= R_{1a}(u)\ \cdots R_{N-1 a}(u)\ R_{Na}(u) \,,  
\label{monodromy}
\end{align}
and the trace in (\ref{transfer}) is over the ``auxiliary'' space, which 
is denoted by $a$. The transfer matrix is engineered to satisfy the 
fundamental commutativity property
\be
\left[ t(u,p) \,, t(v,p) \right] = 0 \hbox{   for all   } u \,, v \,,
\label{commutativity}
\ee
which is the hallmark of integrability. The transfer matrix contains 
the Hamiltonian ($\sim t'(0,p)$, see Appendix \ref{sec:Hamiltonian}) and higher local conserved quantities.

\section{Quantum group symmetry}\label{sec:proof}

We now proceed to show that the transfer matrix (\ref{transfer}) has 
QG symmetry, in accordance with the second column in 
Table  \ref{table:symmetries}.

A key step of our argument is to use a gauge transformation
to bring the right K-matrix ``as close as possible'' to the identity 
matrix. By transforming to this ``unitary'' gauge, the asymptotic (single-row) monodromy matrix
becomes expressed in terms of only the unbroken symmetry generators, which 
then allows us to bring the powerful QISM machinery to bear on the problem.
To this end, we set (see e.g. \cite{Jimbo:1985ua})
\be
\tilde{R}_{12}(u,p) = B_{1}(u,p)\, R_{12}(u)\, B_{1}(-u,p) = 
B_{2}(-u,p)\, R_{12}(u)\, B_{2}(u,p)\,,
\label{gaugeR}
\ee
and \cite{Mezincescu:1990uf}
\begin{align}
\tilde{K}^{R}(u,p) &= B(u,p)\, K^{R}(u,p)\, B(u,p)\,, \non \\
\tilde{K}^{L}(u,p) &= B(-u,p)\, K^{L}(u,p)\, B(-u,p)\,,
\label{gaugeK}
\end{align}
where $B(u,p)$ is a diagonal matrix that maps ${\cal V}$ to itself, 
which we choose as follows
\be
B(u,p) =  \diag \big( \underbrace{e^{\frac{u}{2}}\,, \ldots\,, 
e^{\frac{u}{2}}}_{p}\,, \underbrace{1\,, \ldots\,, 1}_{d-2p}\,, 
\underbrace{e^{-\frac{u}{2}}\,, \ldots\,, 
e^{-\frac{u}{2}}}_{p}\big) \,.
\label{Bgauge}
\ee
Indeed, this gauge transformation brings $K^{R}(u,p)$ (\ref{KR}) to a form with mostly 1's 
on the diagonal
\be
\tilde{K}^{R}(u,p) = \diag \big( \underbrace{1\,, \ldots\,, 
1}_{p}\,, \underbrace{\frac{\gamma e^{u} + 1}{\gamma  + 
e^{u}}\,, \ldots\,, \frac{\gamma e^{u} + 1}{\gamma  + 
e^{u}}}_{d-2p}\,, \underbrace{1\,, \ldots\,, 
1}_{p}\big) \,.
\label{mostlyones}
\ee
For $p=n$, we see that $\tilde{K}^{R}(u,n)$ is exactly equal to
$\id$ if $d=2n$ (i.e., for $A_{2n-1}^{(2)}$, $C_{n}^{(1)}$ and
$D_{n}^{(1)}$); and $\tilde{K}^{R}(u,n)$  differs from $\id$ only in the middle matrix 
element if $d=2n+1$ (i.e., for $A_{2n}^{(2)}$ and $B_{n}^{(1)}$). 

The matrix $B(u,p)$ satisfies
\be
B(u,p) \, B(v,p) = B(u+v,p)\,, \qquad B(0,p) = \id\,,
\ee
as well as 
\be
\left[B_{1}(u,p)\, B_{2}(u,p)\,, R_{12}(v) \right] = 0\,.
\label{BBR}
\ee
With the help of these properties, it can be shown that the 
gauge-transformed R-matrix and K-matrices continue to satisfy their respective 
Yang-Baxter equations.
The crossing symmetry (\ref{crossing}) is also maintained, with \cite{Mezincescu:1990uf}
\be
\tilde{V}(p) = V\, B(\rho,p) = B(-\rho,p)\, V\,,
\label{Vtilde}
\ee
and 
\be
\tilde{M}(p) =\tilde{V}^{t}(p)\, \tilde{V}(p) = B(\rho,p)\, M\, 
B(\rho,p) \,.
\label{Mtilde}
\ee
The transfer matrix (\ref{transfer}) remains invariant under these 
transformations \cite{Mezincescu:1990uf}
\be
t(u,p) = \tr_a \tilde{K}^{L}_{a}(u,p)\, \tilde{T}_a(u,p)\,  
\tilde{K}^{R}_{a}(u,p)\, \widehat{\tilde{T}}_a(u,p) \,,
\label{transfergauge}
\ee
where
\begin{align} 
\tilde{T}_a(u,p) &= \tilde{R}_{aN}(u,p)\ \tilde{R}_{a N-1}(u,p)\ \cdots \tilde{R}_{a1}(u,p) \,,  \non \\
\widehat{\tilde{T}}_a(u,p) &= \tilde{R}_{1a}(u,p)\ \cdots 
\tilde{R}_{N-1 a}(u,p)\ \tilde{R}_{Na}(u,p) \,. 
\label{monodromygauge}
\end{align}

As already remarked in the Introduction, prior to any gauge 
transformation, the R-matrix has the property that
$\check{R}(u) = {\cal
P}R(u)$ commutes with the coproducts of generators of the ``left'' quantum group
$U_{q}(g^{(l)})$ in Table
\ref{table:symmetries} with $p=0$, i.e.\footnote{Further details 
about the generators, coproducts, etc. can be found in Appendix 
\ref{sec:QG}.} 
\be
p=0: \qquad \left[ \check{R}(u)\,, \Delta(H^{(l)}_{j}(0)) \right] = 0 = \left[ 
\check{R}(u)\,, \Delta(E^{\pm\, (l)}_{j}(0)) \right]\,, \qquad j = 1, \ldots, 
n \,.
\ee
In contrast, the gauge-transformed R-matrix given by (\ref{gaugeR}) 
and (\ref{Bgauge}) with $p=n$ has the property that
$\check{\tilde{R}}(u,n) = {\cal P}\tilde{R}(u,n)$ commutes with the
coproducts of generators of the ``right'' quantum group $U_{q}(g^{(r)})$ in Table \ref{table:symmetries} with
$p=n$, i.e. 
\be
p=n: \qquad \left[ \check{\tilde{R}}(u,n)\,, \Delta(H^{(r)}_{j}(n)) \right] = 0 = \left[ 
\check{\tilde{R}}(u,n)\,, \Delta(E^{\pm\, (r)}_{j}(n)) \right]\,, \qquad 
j = 1, \ldots, 
n \,.
\ee
We now use such gauge transformations to prove the QG invariance of 
the open-chain transfer matrix 
$t(u,p)$ for any integer $p \in [0, n]$.

Let us denote by $\tilde{R}^{\pm}(p)$ the asymptotic limits of the 
gauge-transformed R-matrix $\tilde{R}(u,p)$ (\ref{gaugeR})
\be
\tilde{R}^{\pm}(p) = \lim_{u\rightarrow \pm \infty} e^{\mp u} 
\tilde{R}(u,p) \,,
\label{Rtildeasym}
\ee
and we similarly denote by $\tilde{T}^{\pm}_{a}(p)$ the asymptotic limits of the 
gauge-transformed monodromy matrix $\tilde{T}_{a}(u,p)$ 
(\ref{monodromygauge})
\be
\tilde{T}^{\pm}_a(p) = \tilde{R}_{aN}^{\pm}(p)\ \tilde{R}_{a N-1}^{\pm}(p)\\ 
\cdots \tilde{R}_{a1}^{\pm}(p)\ \,.
\label{tildeTpm}
\ee
Let us further denote by $\tilde{T}^{\pm}_{i,j}(p)$ ($1\le i, j \le d$) the matrix elements of 
$\tilde{T}^{\pm}_{a}(p)$ in the auxiliary space, which are operators on the 
quantum space ${\cal V}^{\otimes N}$. 

We show in Appendix \ref{sec:QG} that the operators
$\tilde{T}^{\pm}_{i,j}(p)$ can be expressed in terms of (the quantum enveloping 
algebra of) the unbroken $\hat g$ generators, i.e. the
generators of the quantum groups in the second column of 
Table  \ref{table:symmetries}. Hence, in order to demonstrate the 
QG symmetry of the transfer matrix, it suffices to show 
that
\be
\left[  \tilde{T}^{\pm}_{i,j}(p)\,, t(u,p)\right] = 0 \qquad i, j = 1, 2, \ldots, d
\,.
\label{qgsymmetry}
\ee
To this end, following \cite{Nepomechie:2016ejv} (see also 
\cite{Kulish:1991np, Mezincescu:1991rb}), we first establish several 
lemmas.

\begin{lemma}
\be
\left[\tilde{R}_{12}^{\pm}(p)\,, \tilde{K}^{R}_{2}(u,p)
\right] = 0 \,.
\label{lemma1}
\ee
\label{lemma:1}
\end{lemma}
\noindent
A proof is outlined in Secs. \ref{sec:prooflemma1a} and 
\ref{sec:prooflemma1b}. 

\begin{lemma}
\be
\left[\tilde{R}_{12}^{\pm}(p)\,, \tilde{M}_{1}(p)\, 
\tilde{K}^{L}_{2}(u,p)
\right] = 0 \,.
\label{lemma2}
\ee
\end{lemma}
\noindent

\begin{proof}
We observe that    
\be
\tilde{K}^{L}(u,p) = \tilde{K}^{R}(-u-\rho,p)\, \tilde{M}(p) =  
\tilde{M}(p)\, \tilde{K}^{R}(-u-\rho,p) \,,
\label{KLtilde}
\ee
as follows from (\ref{KL}), (\ref{Mtilde}) and 
(\ref{gaugeK}). Hence,
\begin{align}
\tilde{R}_{12}^{\pm}(p)\, \tilde{M}_{1}(p) \tilde{K}^{L}_{2}(u,p) 
& = \tilde{R}_{12}^{\pm}(p)\, \tilde{M}_{1}(p) \tilde{M}_{2}(p) 
\tilde{K}^{R}_{2}(-u-\rho,p) \non \\
& = \tilde{M}_{1}(p) \tilde{M}_{2}(p) \tilde{R}_{12}^{\pm}(p)\, 
\tilde{K}^{R}_{2}(-u-\rho,p) \non \\
& = \tilde{M}_{1}(p) \tilde{M}_{2}(p) \tilde{K}^{R}_{2}(-u-\rho,p)\, \tilde{R}_{12}^{\pm}(p)  \non \\
& = \tilde{M}_{1}(p) \tilde{K}^{L}_{2}(u,p)\, \tilde{R}_{12}^{\pm}(p) \,,
\end{align}    
where the first and last equalities follow from (\ref{KLtilde}); the second 
equality is a consequence of the fact \cite{Mezincescu:1990uf}
\be
\left[R_{12}(u)\,, M_{1} M_{2} \right] = 0 \,;
\ee
and the third equality follows from Lemma 1 (\ref{lemma1}).
\end{proof}

\begin{lemma}
\be
\left[\tilde{R}_{12}^{\pm}(p)\, \tilde{T}_{1}^{\pm}(p)\,, 
\tilde{T}_{2}(u,p)\, \tilde{K}^{R}_{2}(u,p)\, 
\widehat{\tilde{T}}_{2}(u,p) 
\right] = 0 \,.
\label{lemma3}
\ee
\end{lemma}

\begin{proof}
\quad We recall the gauge-transformed fundamental relation
\be
\tilde{R}_{12}(u_{1}-u_{2},p)\, \tilde{T}_{1}(u_{1},p)\, 
\tilde{T}_{2}(u_{2},p) = \tilde{T}_{2}(u_{2},p)\, 
\tilde{T}_{1}(u_{1},p)\, \tilde{R}_{12}(u_{1}-u_{2},p) \,.
\ee
Taking asymptotic limits of $u_{1}$ yields
\be
\tilde{R}_{12}^{\pm}(p)\, \tilde{T}_{1}^{\pm}(p)\, \tilde{T}_{2}(u,p) 
= \tilde{T}_{2}(u,p)\, 
\tilde{T}_{1}^{\pm}(p)\, \tilde{R}_{12}^{\pm}(p) \,,
\label{RTTasym}
\ee
which further implies
\be
\tilde{T}_{2}^{-1}(u,p)\, \tilde{R}_{12}^{\pm}(p)\,  \tilde{T}_{1}^{\pm}(p) = \tilde{T}_{1}^{\pm}(p)\, 
\tilde{R}_{12}^{\pm}(p)\, \tilde{T}_{2}^{-1}(u,p) \,.
\label{TRTasym}
\ee
Therefore,
\begin{align}
&\tilde{R}_{12}^{\pm}(p)\, \tilde{T}_{1}^{\pm}(p)\, \tilde{T}_{2}(u,p)\, 
\tilde{K}^{R}_{2}(u,p)\, \tilde{T}_{2}^{-1}(-u,p) \non\\
&\qquad = \tilde{T}_{2}(u,p)\, \tilde{T}_{1}^{\pm}(p)\, \tilde{R}_{12}^{\pm}(p)\, 
\tilde{K}^{R}_{2}(u,p)\, \tilde{T}_{2}^{-1}(-u,p) \non\\
&\qquad = \tilde{T}_{2}(u,p)\, \tilde{K}^{R}_{2}(u,p)\, \tilde{T}_{1}^{\pm}(p)\, \tilde{R}_{12}^{\pm}(p)\, 
\tilde{T}_{2}^{-1}(-u,p) \non\\
&\qquad = \tilde{T}_{2}(u,p)\, \tilde{K}^{R}_{2}(u,p)\, 
\tilde{T}_{2}^{-1}(-u,p)\, \tilde{R}_{12}^{\pm}(p)\,  \tilde{T}_{1}^{\pm}(p) \,,
\end{align}
where the first equality follows from (\ref{RTTasym}), the second 
equality follows from Lemma 1 (\ref{lemma1}), and the third 
equality follows from (\ref{TRTasym}). We have therefore demonstrated the 
commutativity property
\be
\left[\tilde{R}_{12}^{\pm}(p)\, \tilde{T}_{1}^{\pm}(p)\,, 
\tilde{T}_{2}(u,p)\, \tilde{K}^{R}_{2}(u,p)\,\tilde{T}_{2}^{-1}(-u,p) 
\right] = 0 \,.
\label{comm}
\ee
Finally, we see from (\ref{monodromygauge}) that
\begin{align}
\tilde{T}_{a}^{-1}(u,p) &= \tilde{R}_{a1}^{-1}(u,p) \cdots \tilde{R}_{aN}^{-1}(u,p) \non \\
&\propto \tilde{R}_{1a}(-u,p) \cdots \tilde{R}_{Na}(-u,p) = 
\widehat{\tilde{T}}_{a}(-u,p)\,,
\end{align}
where the second line follows from unitarity (\ref{unitarity}).
Substituting into (\ref{comm}) we obtain the desired result 
(\ref{lemma3}).
\end{proof}

\begin{lemma}
\be
\tilde{M}^{-1}_{1}(p)\, 
\left((\tilde{R}^{\pm}_{12}(p))^{-1}\right)^{t_{2}}\, \tilde{M}_{1}(p)\, 
\tilde{R}^{\pm\ 
t_{2}}_{12}(p) = \id^{\otimes 2} \,.
\label{lemma4}
\ee
\end{lemma}

\begin{proof}
\quad We write the gauge-transformed unitarity condition 
(\ref{unitarity}) as
\be
\tilde{R}_{12}(u,p)\, \tilde{R}_{12}^{t_{1} t_{2}}(-u,p) = \zeta(u)\, \id^{\otimes 
2}\,,
\label{unitarity2}
\ee
and then use crossing symmetry (\ref{crossing}) to obtain
\be
\tilde{V}_{1}(p)\, \tilde{R}_{12}^{t_{2}}(-u-\rho,p)\, \tilde{V}_{1}(p)\, \tilde{V}_{1}^{t_{1}}(p)\, 
\tilde{R}_{12}^{t_{1}}(u-\rho,p)\, \tilde{V}_{1}^{t_{1}}(p) = \zeta(u)\, \id^{\otimes 2} 
\,,
\label{lemmastep}
\ee
where $\tilde{V}(p)$ is given by (\ref{Vtilde}).
By taking asymptotic limits of (\ref{lemmastep}) and noting that 
$\tilde{V}(p)^{2}= \id$, we obtain
\be
\tilde{R}^{\pm\ t_{2}}_{12}(p)\, \tilde{M}^{-1}_{1}(p)\, \tilde{R}^{\mp\ t_{1}}_{12}(p)\, \tilde{M}_{1}(p) = 
\chi\, \id^{\otimes 2}  \,,
\label{intermed}
\ee
where $\chi$ is given by
\be
\chi = \lim_{u\rightarrow \pm \infty} e^{\mp 2u}\, \zeta(u) = \frac{1}{4}{\delta_1}^{2} \,.
\ee
Moreover, from (\ref{unitarity2}) we obtain 
\be
\tilde{R}^{\pm}_{12}(p)\, \tilde{R}_{12}^{\mp\ t_{1} t_{2}}(p) = 
\chi\, \id^{\otimes 2}\,, 
\ee
which implies that 
\be
\tilde{R}_{12}^{\mp\ t_{1} t_{2}}(p) = \chi\, (\tilde{R}^{\pm}_{12}(p))^{-1} \,, \quad \mbox{  or  } 
\quad
\tilde{R}_{12}^{\mp\ t_{1}}(p) = \chi \left((\tilde{R}^{\pm}_{12}(p))^{-1}\right)^{t_{2}}  \,.
\ee
Substituting into (\ref{intermed}), we obtain
\be
\tilde{R}^{\pm\ t_{2}}_{12}(p)\, \tilde{M}^{-1}_{1}(p)\, \left((\tilde{R}^{\pm}_{12}(p))^{-1}\right)^{t_{2}}\, \tilde{M}_{1}(p) = 
\id^{\otimes 2}  \,,
\ee
which can be rearranged to give the desired result (\ref{lemma4}).  
\end{proof}

We are finally ready to prove the main result (\ref{qgsymmetry}), which 
is equivalent to the following

\begin{prop}
\be
\left[\tilde{T}_{1}^{\pm}(p)\,, t(u,p) \right] = 0 \,.
\label{prop1}
\ee
\end{prop}

\begin{proof}
\quad Recalling that the transfer matrix remains invariant under 
gauge transformations (\ref{transfergauge}), we obtain
\begin{align}
&\tilde{T}_{1}^{\pm}(p)\, t(u,p) \non \\
&\qquad = \tr_{2}\left\{\tilde{T}_{1}^{\pm}(p)\, 
\tilde{K}^{L}_{2}(u,p)\, \tilde{T}_{2}(u,p)\, 
\tilde{K}^{R}_{2}(u,p)\, \widehat{\tilde{T}}_{2}(u,p) \right\} \non \\
&\qquad = \tr_{2}\left\{\tilde{M}^{-1}_{1}(p)\, \tilde{M}_{1}(p)\, 
\tilde{K}^{L}_{2}(u,p)\, (\tilde{R}^{\pm}_{12}(p))^{-1}\, 
\tilde{R}^{\pm}_{12}(p)\, \tilde{T}_{1}^{\pm}(p)\,  \tilde{T}_{2}(u,p)\, 
\tilde{K}^{R}_{2}(u,p)\, \widehat{\tilde{T}}_{2}(u,p) \right\} \non \\
&\qquad = \tr_{2}\left\{\tilde{M}^{-1}_{1}(p)\, (\tilde{R}^{\pm}_{12}(p))^{-1}\, 
\tilde{M}_{1}(p)\, \tilde{K}^{L}_{2}(u,p)\, 
\tilde{T}_{2}(u,p)\, \tilde{K}^{R}_{2}(u,p)\, 
\widehat{\tilde{T}}_{2}(u,p)\, \tilde{R}^{\pm}_{12}(p)\, 
\tilde{T}_{1}^{\pm}(p)   \right\} \non \\
&\qquad = \ldots 
\end{align}
In passing to the third equality, we have used Lemma 2 (\ref{lemma2}) and 
Lemma 3 (\ref{lemma3}). Then
\begin{align}
\ldots &= \tr_{2}\left\{\tilde{M}^{-1}_{1}(p)\, (\tilde{R}^{\pm}_{12}(p))^{-1}\, 
\tilde{M}_{1}(p)\, \tilde{K}^{L}_{2}(u,p)\,  
\tilde{T}_{2}(u,p)\, \tilde{K}^{R}_{2}(u,p)\, 
\widehat{\tilde{T}}_{2}(u,p)\, \tilde{R}^{\pm}_{12}(p) \right\} \tilde{T}_{1}^{\pm}(p)
\non \\
&= \tr_{2}\left\{ A_{12}\, Q_{2}\, \tilde{R}^{\pm}_{12}(p)\, \right\} \tilde{T}_{1}^{\pm}(p)  \non \\
&= \tr_{2}\left\{ A_{12}^{t_{2}}\, \tilde{R}^{\pm\ t_{2}}_{12}(p)\, Q_{2}^{t_{2}}  \right\} \tilde{T}_{1}^{\pm}(p) = \ldots  
\end{align} 
In passing to the second line we have made the identifications
$A_{12} = \tilde{M}^{-1}_{1}(p)\, (\tilde{R}^{\pm}_{12}(p))^{-1}\, \tilde{M}_{1}(p)$
and $Q_{2} = \tilde{K}^{L}_{2}(u,p)\, \tilde{T}_{2}(u,p)\, 
\tilde{K}^{R}_{2}(u,p)\, \widehat{\tilde{T}}_{2}(u,p)$. Finally, we obtain
\begin{align}
\ldots &= \tr_{2}\left\{ \tilde{M}^{-1}_{1}(p)\, 
\left((\tilde{R}^{\pm}_{12}(p))^{-1}\right)^{t_{2}}\, \tilde{M}_{1}(p)\, \tilde{R}^{\pm\ 
t_{2}}_{12}(p)\,  Q_{2}^{t_{2}} \right\} \tilde{T}_{1}^{\pm}(p) \non \\
&=  \tr_{2}\left\{ Q_{2}^{t_{2}} \right\} \tilde{T}_{1}^{\pm}(p) \non \\
&= t(u,p)\,  \tilde{T}_{1}^{\pm}(p) \,.
\end{align}
In passing to the second line we have used Lemma 4
(\ref{lemma4}); and we have used (\ref{transfergauge}) again to pass 
to the third line.
\end{proof}

\section{Duality symmetry}\label{sec:duality}

We now show that the transfer matrix $t(u,p)$ 
(\ref{transfer}) for the cases $C_{n}^{(1)}$ and $D_{n}^{(1)}$
has a $p \leftrightarrow n-p$ ``duality'' symmetry. 
In order to prove the general result (\ref{duality}), we need the 
following lemma:

\begin{lemma}
The R-matrices for both $C_{n}^{(1)}$ and $D_{n}^{(1)}$ obey
\begin{align}
U_{1}\, R_{12}(u)\, U_{1} &= W_{2}^{t}(u)\, R_{12}(u)\, 
W_{2}^{t}(u) \,, \non \\
U_{2}\, R_{12}(u)\, U_{2} &= W_{1}(u)\, R_{12}(u)\, 
W_{1}(u) \,,
\label{RmatpropCD}
\end{align}
where $U$ and $W(u)$ are the following $(2n) \times (2n)$ matrices
\begin{align}
U &= \left(\begin{array}{c|c}
0  &  \id_{n \times n} \\
\hline
\id_{n \times n}  & 0 
\end{array}\right)_{2n \times 2n} \,, & U^{2} &= \id \,, \non \\
W(u) &= \left(\begin{array}{c|c}
0  &  e^{-\frac{u}{2}}\id_{n \times n} \\
\hline
e^{\frac{u}{2}}\id_{n \times n}  & 0 
\end{array}\right)_{2n \times 2n} \,, & W(u)^{2} &= \id \,.
\label{UWmats}
\end{align}    
Furthermore, the K-matrices (\ref{KR}) and (\ref{KL}) obey
\begin{align}
W(u)\, K^{R}(u,p)\, W^{t}(u) &= f^{R}(u,p)\,  K^{R}(u,n-p) \non \\
W^{t}(u)\, K^{L}(u,p)\, W(u) &= f^{L}(u,p)\,  K^{L}(u,n-p) \,,
\label{dualityK}
\end{align}
where $f^{R}(u,p)$ and $f^{L}(u,p)$ are scalar functions given by 
\begin{align}
f^{R}(u,p) &= \frac{\gamma_{0}\, e^{u} + e^{(2n-4p)\eta}}{\gamma_{0} + e^{u+(2n-4p)\eta}} \,,  \non \\
f^{L}(u,p) &= 
\frac{\gamma_{0}\, e^{u} + e^{(4p+2n \pm 4)\eta}}{\gamma_{0}\, e^{(4n \pm 4)\eta} + 
e^{u+(4p-2n)\eta}}  \mbox{ with } \left\{ \begin{array}{ll}
+ & \mbox{ for } C_{n}^{(1)}  \\
- & \mbox{ for } D_{n}^{(1)}  
\end{array} \right.
\,,
\end{align}
where $\gamma_{0}=\pm 1$ is a parameter appearing in the K-matrix, see (\ref{gamma}).
\label{lemma:5}
\end{lemma}
A proof of (\ref{RmatpropCD}) is outlined in Sec. 
\ref{sec:proofdualitylemma}. 

The main duality result is given by the following proposition:
\begin{prop}
For the cases $C_{n}^{(1)}$ and $D_{n}^{(1)}$, the transfer matrix 
has the duality symmetry
\be
{\cal U}\, t(u,p)\, {\cal U} = f(u,p)\, t(u,n-p)  \,, 
\label{duality}
\ee 
where ${\cal U}$ is the quantum-space operator 
\be
{\cal U} = U_{1}\ldots U_{N} \,, \qquad {\cal U}^{2} = \id^{\otimes 
N}\,,
\label{Ucal}
\ee
and the scalar factor $f(u,p)$ is given by 
\be
f(u,p) = f^{L}(u,p)\, f^{R}(u,p) \,.
\ee 
\end{prop}

\begin{proof}
\quad 
We see from (\ref{RmatpropCD}) that the monodromy matrices (\ref{monodromy}) transform as 
follows
\begin{align}
{\cal U}\, T_{a}(u)\, {\cal U} &= W_{a}(u)\, T_{a}(u)\, W_{a}(u) \,, \non \\
{\cal U}\, \widehat{T}_{a}(u)\, {\cal U} &= W^{t}_{a}(u)\, 
\widehat{T}_{a}(u)\, W^{t}_{a}(u) \,.
\label{dualitymonod}
\end{align}
Evaluating ${\cal U}\, t(u,p)\, {\cal U}$ using the definition (\ref{transfer}) of the 
transfer matrix together with (\ref{dualitymonod}) and 
(\ref{dualityK}), we arrive at the desired result (\ref{duality}).
\end{proof}
A similar duality symmetry was noted for the case $A_{n-1}^{(1)}$ in 
\cite{Doikou:1998ek}.

As a consequence of the duality symmetry (\ref{duality}), for each eigenvalue 
$\Lambda(u,p)$ of $t(u,p)$, there is a corresponding 
eigenvalue $\Lambda(u,n-p)$ of $t(u,n-p)$ such that 
\be
\Lambda(u,p)  = f(u,p)\, \Lambda(u,n-p) \,.
\label{duality2}
\ee

\subsection{Action of duality on the QG generators}

For the case $C_{n}^{(1)}$, the transfer matrix $t(u,p)$ has the symmetry 
$U_{q}(C_{n-p}) \otimes U_{q}(C_{p})$ (see again Table 
\ref{table:symmetries}), while $t(u,n-p)$ (its image under the 
duality transformation (\ref{duality})) has the symmetry 
$U_{q}(C_{p}) \otimes U_{q}(C_{n-p})$. Under a duality 
transformation, the generators of the ``left'' 
symmetry factor of $t(u,p)$ (namely, $U_{q}(C_{n-p})$)  are mapped to 
the generators of the ``right''  symmetry factor of $t(u,n-p)$ 
(which is also $U_{q}(C_{n-p})$). Similarly, the generators of the ``right'' 
symmetry factor of $t(u,p)$ (namely, $U_{q}(C_{p})$)  are mapped to 
the generators of the ``left''  symmetry factor of $t(u,n-p)$ 
(which is also $U_{q}(C_{p})$). The case $D_{n}^{(1)}$ is identical, except with 
$D$'s replacing the $C$'s. In other words,
\begin{align}
U\, H_{i}^{(l)}(p)\, U &=   H_{i}^{(r)}(n-p)\,, &
U\, E_{i}^{\pm\, (l)}(p)\, U &=   E_{i}^{\pm\, (r)}(n-p)\,,
\qquad i = 1, 2, \ldots, n-p \,, \non \\
U\, H_{i}^{(r)}(p)\, U &=   H_{i}^{(l)}(n-p)\,, &
U\, E_{i}^{\pm\, (r)}(p)\, U &=   E_{i}^{\pm\, (l)}(n-p)\,,
\qquad i = 1, 2, \ldots, p \,.
\label{dualitygenerators}
\end{align}
and similarly for the coproducts. In order to obtain the general 
result (\ref{prop3}), we need a few lemmas.

\begin{lemma}
\be
W^{t}(u) = B(u,n-p)\, U\, B(-u,p) \,,
\label{lemma5}
\ee
where $U$ and $W(u)$ are given by (\ref{UWmats}).
\end{lemma}

\begin{proof}
We evaluate the RHS by writing all three matrices in terms of $n \times n$ blocks:
\begin{align}
&RHS  \non\\
&=  \left(\begin{array}{cc|cc}
e^{\frac{u}{2}}\id_{(n-p) \times (n-p)} & & & \\
& \id_{p \times p} & & \\
\hline
& & \id_{p \times p} &  \\
& & & e^{-\frac{u}{2}}\id_{(n-p) \times (n-p)}
\end{array}\right) 
\left(\begin{array}{c|c}
  &  \id_{n \times n} \\
\hline
\id_{n \times n}  &  
\end{array}\right)\, B(-u,p) \non\\
&= \left(\begin{array}{cc|cc}
& & e^{\frac{u}{2}}\id_{(n-p) \times (n-p)} & \\
& & & \id_{p \times p} \\
\hline
\id_{p \times p} &  & & \\
& e^{-\frac{u}{2}}\id_{(n-p) \times (n-p)} 
\end{array}\right)
\left(\begin{array}{cc|cc}
e^{-\frac{u}{2}}\id_{p \times p} & & & \\
& \id_{(n-p) \times (n-p)} & & \\
\hline
& & \id_{(n-p) \times (n-p)} &  \\
& & & e^{\frac{u}{2}}\id_{p \times p}
\end{array}\right) \non\\
&= \left(\begin{array}{c|c}
  &  e^{\frac{u}{2}}\id_{n \times n} \\
\hline
e^{-\frac{u}{2}}\id_{n \times n}  &  
\end{array}\right) = LHS \,.
\end{align}
\end{proof}

\begin{lemma}
The gauge-transformed R-matrices for $C_{n}^{(1)}$ and $D_{n}^{(1)}$ obey
\be
U_{1}\, \tilde{R}_{12}(u,p)\, U_{1} = U_{2}\, \tilde{R}_{12}(u,n-p)\, 
U_{2} \,.
\label{lemma6}
\ee
\end{lemma}

\begin{proof}   
Recalling the definition of the gauge-transformed R-matrix (\ref{gaugeR}), we see that 
\begin{align}
U_{1}\, \tilde{R}_{12}(u,p)\, U_{1} &= U_{1}\, B_{2}(-u,p)\, 
R_{12}(u)\, B_{2}(u,p)\, U_{1} \non \\
&= B_{2}(-u,p)\, U_{1}\, 
R_{12}(u)\, U_{1}\, B_{2}(u,p) \non \\
&= B_{2}(-u,p)\,  W^{t}_{2}(u)\,
R_{12}(u)\,  W^{t}_{2}(u)\, B_{2}(u,p) \non \\
&= U_{2}\, B_{2}(-u,n-p)\,
R_{12}(u)\, B_{2}(u,n-p)\, U_{2} \non \\
&= U_{2}\, \tilde{R}_{12}(u,n-p)\, U_{2} \,.
\label{lemma6proof}
\end{align}
In passing to the third line, we have used the result 
(\ref{RmatpropCD}); in passing to the fourth line, we use
\be
B(-u,p)\, W^{t}(u)  = U\, B(-u,n-p)\,, \qquad  W^{t}(u)\, B(u,p)  = B(u,n-p)\, U\,,
\ee
which follow from (\ref{lemma5}); and we pass to the last line 
using again the definition of the gauge-transformed R-matrix.
\end{proof}

\begin{lemma}
The gauge-transformed monodromy matrices for $C_{n}^{(1)}$ and $D_{n}^{(1)}$
transform under duality as 
\be
{\cal U}\, \tilde{T}_{a}(u,p)\, {\cal U} = U_{a}\, \tilde{T}_{a}(u,n-p)\, 
U_{a} \,,
\label{lemma7}
\ee
where ${\cal U}$ is given by (\ref{Ucal}).
\end{lemma}
\begin{proof} 
This result follows immediately from the definition of 
$\tilde{T}_{a}(u,p)$ (\ref{monodromygauge}) and the result 
(\ref{lemma6}).
\end{proof}     

Finally, taking asymptotic limits of the result (\ref{lemma7}), we 
obtain the sought-after result:

\begin{prop}
For the cases $C_{n}^{(1)}$ and $D_{n}^{(1)}$, the asymptotic 
gauge-transformed monodromy matrices
$\tilde{T}^{\pm}_{a}(p)$ transform under duality as 
\be
{\cal U}\, \tilde{T}^{\pm}_{a}(p)\, {\cal U} = U_{a}\, \tilde{T}^{\pm}_{a}(n-p)\, 
U_{a} \,.
\label{prop3}
\ee
\end{prop}
\noindent
From the result (\ref{prop3}), we can read off the 
transformation properties of the coproducts of 
the QG generators under duality, thereby generalizing (\ref{dualitygenerators}).

\subsection{Self-duality}

For $p=\frac{n}{2}$ with $n$ even, we see that the duality relation 
(\ref{duality}) implies that the transfer matrix is self-dual
\be
\left[ {\cal U}\,,  t(u, \tfrac{n}{2}) \right] = 0 \,,
\label{selfdual}
\ee
since $f(u,\frac{n}{2}) = 1$. 
This self-duality symmetry maps the ``left'' and ``right'' generators into each other
\begin{align}
U\, H_{i}^{(l)}(\tfrac{n}{2})\, U &=   H_{i}^{(r)}(\tfrac{n}{2})\,, &
U\, E_{i}^{\pm\, (l)}(\tfrac{n}{2})\, U &=   E_{i}^{\pm\, (r)}(\tfrac{n}{2})\,,
\qquad i = 1, 2, \ldots, \tfrac{n}{2} \,, 
\end{align}
as follows from (\ref{dualitygenerators}). Hence, this symmetry 
maps the representations $(1, {\bf R})$ and $({\bf R}, 1)$ (i.e., 
with ``left'' and ``right'' singlets, respectively) into each other;
and therefore these states are degenerate (i.e., have the same 
transfer-matrix eigenvalue). This degeneracy is discussed further
in Section \ref{sec:degen}.

\subsubsection{Bonus symmetry for $\gamma_{0}=-1$}\label{sec:bonus} 

For the self-dual cases (namely, $C_{n}^{(1)}$ and $D_{n}^{(1)}$ with $p=\frac{n}{2}$ 
and $n$ even) with $\gamma_{0}=-1$, there is an additional (``bonus'') 
symmetry, which leads to even higher degeneracies for the 
transfer-matrix eigenvalues.

In order to exhibit this symmetry, it is convenient to introduce the 
matrix $\bar{U}$, which is similar to the duality matrix $U$ (\ref{UWmats}),
\be
\bar{U} = \left(\begin{array}{cc|cc}
  &  & i\id_{\frac{n}{2} \times \frac{n}{2}} & \\
& & &  -i\id_{\frac{n}{2} \times \frac{n}{2}} \\
\hline
-i\id_{\frac{n}{2} \times \frac{n}{2}} &  &  \\
& i\id_{\frac{n}{2} \times \frac{n}{2}} &  &  
\end{array}\right)_{2n \times 2n} \,, \qquad  \bar{U}^{2} = \id \,, 
\label{barUmat}
\ee
and which satisfies
\be
\bar{U} U = - U \bar{U} = i D\,,
\label{prodbarUU}
\ee
where $D$ is the diagonal matrix
\be
D = \diag \big( \underbrace{1\,, \ldots\,, 1}_{\frac{n}{2}}\,, \underbrace{-1\,, \ldots\,, -1}_{n}\,,
\underbrace{1\,, \ldots\,, 1}_{\frac{n}{2}}\big) \,.
\label{Dmat}
\ee
Similarly to (\ref{lemma6}), we find that the gauge-transformed 
R-matrix obeys
\be
\bar{U}_{1}\, \tilde{R}_{12}(u,\tfrac{n}{2})\, \bar{U}_{1} = 
\bar{U}_{2}\, \tilde{R}_{12}(u,\tfrac{n}{2})\, 
\bar{U}_{2} \,,
\label{UbartildeR}
\ee
as well as 
\be
D_{1}\, \tilde{R}_{12}(u,\tfrac{n}{2})\, D_{1} = 
D_{2}\, \tilde{R}_{12}(u,\tfrac{n}{2})\, 
D_{2} \,.
\label{DtildeR}
\ee
Moreover, the gauge-transformed right K-matrix (\ref{mostlyones}) is 
equal to $D$ \footnote{We emphasize that the result (\ref{tildeKspecial}) holds only for 
$\gamma_{0}=-1$, and we assume that $\gamma_{0}=-1$ in the remainder 
of this subsection.}
\be
\tilde{K}^{R}(u,\tfrac{n}{2}) = D \,.
\label{tildeKspecial}
\ee
It follows from the BYBE (\ref{BYBEm}) that
\be
\tilde{R}_{12}(u - v,\tfrac{n}{2})\, D_1\ \tilde{R}_{21} (u + v,\tfrac{n}{2})\, D_2
= D_2\, \tilde{R}_{12}(u + v,\tfrac{n}{2})\, D_1\, \tilde{R}_{21}(u - 
v,\tfrac{n}{2})  \,.
\label{BYBEresult}
\ee

The key result is given by the following proposition
 
\begin{prop}
For the cases $C_{n}^{(1)}$ and $D_{n}^{(1)}$ with $p=\frac{n}{2}$ 
($n$ even) and $\gamma_{0}=-1$, the transfer matrix has the bonus symmetry    
\be
\left[{\cal \bar{U}}\,, t(u,\tfrac{n}{2}) \right] = 0 \,,
\label{propbonus}
\ee
where ${\cal \bar{U}}$ is the quantum-space operator given by 
\footnote{Note that ${\cal \bar{U}}$ contains only one factor of
$\bar{U}$; all the other factors are $U$.}
\be
{\cal \bar{U}} = \bar{U}_{1} U_{2} \cdots U_{N} \,, \qquad {\cal 
\bar{U}}^{\, 2} = \id^{\otimes N}\,.
\label{calbarU}
\ee
\end{prop}

\begin{proof}
We see from (\ref{UbartildeR}) that the gauge-transformed monodromy matrices (\ref{monodromygauge}) transform as 
follows
\begin{align}
{\cal \bar{U}}\, \tilde{T}_{a}(u,\tfrac{n}{2})\, {\cal \bar{U}} &= 
-i\, U_{a} \tilde{R}_{a N}(u,\tfrac{n}{2})\, \tilde{R}_{a, 
N-1}(u,\tfrac{n}{2})\, \cdots \tilde{R}_{a 2}(u,\tfrac{n}{2})\, 
D_{a}\, \tilde{R}_{a 1}(u,\tfrac{n}{2})\, {\bar U}_{a}
 \,, \non \\
{\cal \bar{U}}\,\widehat{\tilde{T}}_{a}(u,\tfrac{n}{2})\, {\cal 
\bar{U}} &= i\, {\bar U}_{a}\, \tilde{R}_{1 a}(u,\tfrac{n}{2})\, 
D_{a}\, \tilde{R}_{2 a}(u,\tfrac{n}{2})\, \tilde{R}_{3 
a}(u,\tfrac{n}{2})\,  \cdots \tilde{R}_{N 
a}(u,\tfrac{n}{2})\, U_{a}\,,
\end{align}   
where we have also used (\ref{prodbarUU}). Starting from the gauge-transformed 
expression for the transfer matrix (\ref{transfergauge}), and also making 
use of (\ref{tildeKspecial}), we obtain
\begin{eqnarray}
\lefteqn{{\cal \bar{U}}\, t(u,\tfrac{n}{2})\, {\cal \bar{U}} }\non \\
&&= \tr_{a} 
\tilde{K}^{L}_{a}(u, \tfrac{n}{2})\, 
\tilde{R}_{a N}(u,\tfrac{n}{2})\, \cdots \tilde{R}_{a 2}(u,\tfrac{n}{2})\, 
D_{a}\, \tilde{R}_{a 1}(u,\tfrac{n}{2})\, D_{a} \, 
\tilde{R}_{1 a}(u,\tfrac{n}{2})\, 
D_{a}\, \tilde{R}_{2 a}(u,\tfrac{n}{2})\,  \cdots \tilde{R}_{N 
a}(u,\tfrac{n}{2}) \non \\
&&= \tr_{a} \tilde{K}^{L}_{a}(u, \tfrac{n}{2})\, 
\tilde{R}_{a N}(u,\tfrac{n}{2})\, \cdots \tilde{R}_{a 1}(u,\tfrac{n}{2})\, 
D_{a}\, 
\tilde{R}_{1 a}(u,\tfrac{n}{2})\,   \cdots \tilde{R}_{N 
a}(u,\tfrac{n}{2}) \non \\
&&= \tr_{a} \tilde{K}^{L}_{a}(u, \tfrac{n}{2})\,  \tilde{T}_{a}(u,\tfrac{n}{2})\, 
D_{a}\, \widehat {\tilde{T}}_{a}(u,\tfrac{n}{2})  \non \\
&&= t(u,\tfrac{n}{2})\,,
\end{eqnarray}
which implies the desired result (\ref{propbonus}). In passing to the 
first equality, we have used the cyclic property of the trace, and 
the fact $U_{a}\, \tilde{K}^{L}_{a}(u, \tfrac{n}{2})\, U_{a} = - 
\tilde{K}^{L}_{a}(u, \tfrac{n}{2})$; and to pass to
the second equality, we have used the result
\be
D_{a}\, \tilde{R}_{a 1}(u,\tfrac{n}{2})\, D_{a} \, 
\tilde{R}_{1 a}(u,\tfrac{n}{2})\, D_{a} = \tilde{R}_{a 1}(u,\tfrac{n}{2})\, D_{a} \, 
\tilde{R}_{1 a}(u,\tfrac{n}{2}) \,,
\ee
which follows from (\ref{BYBEresult}).  
\end{proof}

Recalling the definitions of ${\cal U}$ (\ref{Ucal}) and ${\cal 
\bar{U}}$ (\ref{calbarU}) as well as the property 
(\ref{prodbarUU}), it is easy to see that
\be
\left[ {\cal U}\,, {\cal \bar{U}} \right] = -2i\, {\cal D} \,,
\label{commcalUcalbarU}
\ee
where ${\cal D}$ is the quantum-space operator defined by
\be
{\cal D} = D_{1} = D \otimes \id^{\otimes (N-1)} \,, \qquad  {\cal 
D}^{2} =  \id^{\otimes N}\,.
\label{Dop}
\ee 
The fact that ${\cal D}$ commutes with the transfer matrix is now a 
simple corollary of (\ref{propbonus}):

\begin{corollary*}
    For the cases $C_{n}^{(1)}$ and $D_{n}^{(1)}$ with $p=\frac{n}{2}$ 
($n$ even) and $\gamma_{0}=-1$, the transfer matrix commutes with the 
operator ${\cal D}$ (\ref{Dop})
\be
    \left[ {\cal D}\,,  t(u, \tfrac{n}{2}) \right] = 0 \,.
\label{calDsymmetry}
\ee
\end{corollary*}

\begin{proof}
    Using (\ref{commcalUcalbarU}) and the Jacobi identity, we see that
\begin{align}
\left[ {\cal D}\,,  t(u, \tfrac{n}{2}) \right] &= \frac{i}{2} \left[  
\left[ {\cal U}\,, {\cal \bar{U}} \right] \,, t(u, \tfrac{n}{2}) 
\right] \non \\
&= -\frac{i}{2} \left[  
\left[ {\cal \bar{U}}\,, t(u, \tfrac{n}{2}) \right] \,, {\cal U}
\right] 
-\frac{i}{2} \left[  
\left[ t(u, \tfrac{n}{2})  \,, {\cal U} \right] \,, {\cal \bar{U}}
\right] \non \\
&=0 \,,
\end{align}
where the final equality follows from the symmetries (\ref{selfdual}) and 
(\ref{propbonus}).
\end{proof}

The symmetry (\ref{calDsymmetry}) gives rise to additional degeneracies of the 
transfer-matrix eigenvalues. Indeed, let $|\Lambda\rangle$ be a 
simultaneous eigenket of the transfer matrix and of the self-duality 
operator ${\cal U}$,
\begin{align}
t(u, \tfrac{n}{2})\, |\Lambda\rangle &= \Lambda(u, \tfrac{n}{2})\, 
|\Lambda\rangle \,, \non \\
{\cal U}\, |\Lambda\rangle &= \mu \, 
|\Lambda\rangle \,, \qquad \mu = \pm 1 \,.
\end{align}
Since ${\cal U}$ and ${\cal D}$ do not commute\footnote{Indeed,
$\left[ {\cal U}\,, {\cal D} \right] =  \left[ U\,, D \right] \otimes 
U \otimes \cdots \otimes U = 2 i\, {\cal \bar{U}}$, see (\ref{prodbarUU}).}
, $|\Lambda\rangle$ is 
not necessarily an eigenket of ${\cal D}$, in which case ${\cal D}\, 
|\Lambda\rangle$ is a linearly independent eigenket with the same 
transfer-matrix eigenvalue $\Lambda(u, \tfrac{n}{2})$ as $|\Lambda\rangle$.
Note that $|\Lambda\rangle$ necessarily belongs to a QG representation of the 
form $({\bf R}, {\bf R})$ or $({\bf R}_{1}, {\bf R}_{2}) 
\oplus ({\bf R}_{2}, {\bf R}_{1})$; hence, the bonus symmetry implies 
the existence of a second set of states of the form $({\bf R}, {\bf R})$ or $({\bf R}_{1}, {\bf R}_{2}) 
\oplus ({\bf R}_{2}, {\bf R}_{1})$. In particular, the degeneracy of the corresponding 
transfer-matrix eigenvalue becomes {\em doubled} as a consequence of 
the bonus symmetry.

\section{$Z_{2}$ symmetries}\label{sec:Z2}

We now show that the transfer matrix $t(u,p)$ (\ref{transfer}) has
a discrete ``right'' $Z_{2}$ symmetry that maps complex representations of
$U_{q}(D_{p})$ to their conjugates for the cases $A_{2n-1}^{(2)}$,
$B_{n}^{(1)}$ and $D_{n}^{(1)}$; and, for the latter case, there is an
additional ``left'' $Z_{2}$ symmetry that maps complex representations of
$U_{q}(D_{n-p})$ to their conjugates. We shall see in Section 
\ref{sec:degen} that these discrete
symmetries give rise to degeneracies in the spectrum beyond those 
expected from QG symmetry.\footnote{The $Z_{2}$ symmetry 
for the case $A_{2n-1}^{(2)}$ with $p=n$ was conjectured in \cite{Nepomechie:2017hgw}.}

\subsection{The ``right'' $Z_{2}$}

In order to prove the main result 
(\ref{Z2rightprop}), we need the following lemma:

\begin{lemma}
The R-matrices for $A_{2n-1}^{(2)}$, $B_{n}^{(1)}$ and $D_{n}^{(1)}$ obey
\begin{align}
Z^{(r)}_{1}\, R_{12}(u)\, Z^{(r)}_{1} &= Y^{t}_{2}(u)\, 
R_{12}(u)\, Y^{t}_{2}(u)\,, \non \\
Z^{(r)}_{2}\, R_{12}(u)\, Z^{(r)}_{2} &= Y_{1}(u)\, 
R_{12}(u)\, Y_{1}(u)\,,
\label{Z2rightR}
\end{align}
where $Z^{(r)}$ and $Y(u)$ are the following $d \times d$ matrices
\begin{align}
Z^{(r)} &= \left(\begin{array}{ccc}
0 & 0 & 1 \\
0 & \id_{(d-2) \times (d-2)} & 0 \\
1 & 0 & 0
\end{array}\right)_{d \times d} \,, & Z^{(r)\, 2} &= \id \,, \non \\
Y(u) &= \left(\begin{array}{ccc}
0 & 0 & e^{-u} \\
0 & \id_{(d-2) \times (d-2)} & 0 \\
e^{u} & 0 & 0
\end{array}\right)_{d \times d} \,, & Y(u)^{2} &= \id \,. 
\label{Z2rightmat}
\end{align}
Furthermore, for $p>0$, the K-matrices (\ref{KR}) and (\ref{KL}) obey
\begin{align}
Y(u)\, K^{R}(u,p)\, Y^{t}(u) &=  K^{R}(u,p) \,, \non \\
Y^{t}(u)\, K^{L}(u,p)\, Y(u) &=  K^{L}(u,p) \,.
\label{Z2rightK}
\end{align}
\label{lemma:9}
\end{lemma}   
A proof of (\ref{Z2rightR}) is outlined in Sec. 
\ref{sec:proofz2rlemma}. 

The main result concerning the ``right'' $Z_{2}$ symmetry
is contained in the following proposition:
    
\begin{prop}
For the cases $A_{2n-1}^{(2)}$, $B_{n}^{(1)}$ and $D_{n}^{(1)}$ with 
$p>0$,  the transfer matrix has the ``right'' $Z_{2}$ symmetry  
\be
\left[{\cal Z}^{(r)}\,, t(u,p) \right] = 0 \,, 
\label{Z2rightprop}
\ee 
where ${\cal Z}^{(r)}$ is the quantum-space operator 
\be
{\cal Z}^{(r)} = Z_{1}^{(r)}\ldots Z_{N}^{(r)} \,, \qquad {\cal 
Z}^{(r)\, 2} = \id^{\otimes N}\,.
\label{Zcalright}
\ee
\end{prop}

\begin{proof}
\quad We see from (\ref{Z2rightR}) that the monodromy matrices (\ref{monodromy}) transform as 
follows
\begin{align}
{\cal Z}^{(r)}\, T_{a}(u)\, {\cal Z}^{(r)} &= Y_{a}(u)\,T_{a}(u)\, 
Y_{a}(u) \,, \non \\
{\cal Z}^{(r)}\, \widehat{T}_{a}(u)\, {\cal Z}^{(r)} &= Y_{a}^{t}(u)\, 
\widehat{T}_{a}(u)\, Y_{a}^{t}(u) \,.
\label{Z2rightmonod}
\end{align}
Evaluating ${\cal Z}^{(r)}\, t(u,p)\, {\cal Z}^{(r)}$ using the definition (\ref{transfer}) of the 
transfer matrix together with (\ref{Z2rightmonod}) and  (\ref{Z2rightK}), we arrive at the 
result (\ref{Z2rightprop}).
\end{proof}

\subsubsection{Action of the ``right'' $Z_{2}$ on the QG generators}

In order to determine the action of the ``right'' $Z_{2}$ on the 
QG generators, we use a set of lemmas that are analogous to 
(\ref{lemma5}), (\ref{lemma6}) and (\ref{lemma7}), and which have similar 
proofs:

\begin{lemma}
\be
Y(u) = B(-u,p)\, Z^{(r)}\, B(u,p) \,, \qquad  p>0\,,
\label{lemma8}
\ee
where $Z^{(r)}$ and $Y(u)$ are given by (\ref{Z2rightmat}).
\end{lemma}

\begin{lemma}
The gauge-transformed R-matrices (\ref{gaugeR}) for $A_{2n-1}^{(2)}$, 
$B_{n}^{(1)}$ and $D_{n}^{(1)}$ with $p>0$ obey
\be
Z^{(r)}_{1}\, \tilde{R}_{12}(u,p)\, Z^{(r)}_{1} = Z^{(r)}_{2}\, 
\tilde{R}_{12}(u,p)\, Z^{(r)}_{2}\,.
\label{Z2righttildeR}
\ee
\end{lemma}

\begin{lemma}
The gauge-transformed monodromy matrices (\ref{monodromygauge}) for 
$A_{2n-1}^{(2)}$, $B_{n}^{(1)}$ and $D_{n}^{(1)}$ with $p>0$ 
transform under the ``right'' $Z_{2}$ as 
\be
{\cal Z}^{(r)}\, \tilde{T}_{a}(u,p)\, {\cal Z}^{(r)} = Z_{a}^{(r)}\, 
\tilde{T}_{a}(u,p)\, Z_{a}^{(r)} \,.
\label{Z2rightmonodtilde}
\ee
\end{lemma}

Finally, taking asymptotic limits of the result (\ref{Z2rightmonodtilde}), we 
obtain the sought-after result:

\begin{prop}
For the cases $A_{2n-1}^{(2)}$, $B_{n}^{(1)}$ and $D_{n}^{(1)}$ with $p>0$, 
the asymptotic gauge-transformed monodromy matrices $\tilde{T}^{\pm}_{a}(p)$ transform under the ``right'' $Z_{2}$ as 
\be
{\cal Z}^{(r)}\, \tilde{T}^{\pm}_{a}(p)\, {\cal Z}^{(r)} = 
Z_{a}^{(r)}\, \tilde{T}^{\pm}_{a}(p)\, Z_{a}^{(r)} \,.
\label{prop5}
\ee
\end{prop}

We can read off from this result how the coproducts of 
the ``right'' QG generators transform under this $Z_{2}$ symmetry. 
In particular, we observe that
\begin{align}
Z_{a}^{(r)}\, H^{(r)}_{j}\, Z_{a}^{(r)}  &= \left\{ 
\begin{array}{rl}
    H^{(r)}_{j} & \mbox{ for }  j = 1, \ldots, p-1\,, \\
   -H^{(r)}_{p} & \mbox{ for }  j=p 
   \end{array} \right. \,, \non \\
   Z_{a}^{(r)}\, E^{\pm (r)}_{j}\, Z_{a}^{(r)} &= \left\{   
\begin{array}{ll}
     E^{\pm (r)}_{j} & \mbox{ for }  j = 1, \ldots, p-2\,, \\
     E^{\pm (r)}_{p} & \mbox{ for }  j=p-1 
    \end{array} \right. \,.
\end{align}
Hence, this transformation maps complex representations of $U_{q}(D_{p})$ to 
their conjugates.

\subsection{The ``left'' $Z_{2}$}

In order to prove the main result (\ref{Z2leftprop}), we need the 
following lemma:

\begin{lemma}
The R-matrix for $D_{n}^{(1)}$ obeys
\be
Z^{(l)}_{1}\, R_{12}(u)\, Z^{(l)}_{1} = Z^{(l)}_{2}\, 
R_{12}(u)\, Z^{(l)}_{2}\,,
\label{Z2leftR}
\ee
where $Z^{(l)}$ is the following $2n \times 2n$ matrix
\be
Z^{(l)} = \left(\begin{array}{cccc}
\id_{(n-1) \times (n-1)} &  &  &  \\
& 0 & 1 &  \\
& 1 & 0 &  \\
& & & \id_{(n-1) \times (n-1)}
\end{array}\right)_{2n \times 2n} \,, \qquad Z^{(l)\, 2} = \id \,.
\label{Z2leftmat}
\ee
Furthermore, for $p<n$, the K-matrices (\ref{KR}) and (\ref{KL}) obey
\begin{align}
Z^{(l)}\, K^{R}(u,p)\, Z^{(l)} &=  K^{R}(u,p) \,, \non \\
Z^{(l)}\, K^{L}(u,p)\, Z^{(l)} &=  K^{L}(u,p) 
\,.
\label{Z2leftK}
\end{align}
\label{lemma:13}
\end{lemma}
A proof of (\ref{Z2leftR}) is outlined in Sec. 
\ref{sec:proofz2llemma}. 

The main result concerning the ``left'' $Z_{2}$ symmetry is contained in the following proposition:

\begin{prop}
For the case $D_{n}^{(1)}$ with $p<n$, the transfer matrix has the ``left'' $Z_{2}$ symmetry
\be
\left[{\cal Z}^{(l)}\,, t(u,p) \right] = 0 \,, 
\label{Z2leftprop}
\ee 
where ${\cal Z}^{(l)}$ is the quantum-space operator
\be
{\cal Z}^{(l)} = Z_{1}^{(l)}\ldots Z_{N}^{(l)} \,, \qquad {\cal Z}^{(l)\, 2} = \id^{\otimes N} \,.
\label{Zcalleft}
\ee
\end{prop}

\begin{proof}
\quad We see from (\ref{Z2leftR}) that the monodromy matrices (\ref{monodromy}) transform as follows
\begin{align}
{\cal Z}^{(l)}\, T_{a}(u)\, {\cal Z}^{(l)} &= Z_{a}^{(l)}\, 
T_{a}(u)\, Z_{a}^{(l)} \,, \non \\
{\cal Z}^{(l)}\, \widehat{T}_{a}(u)\, {\cal Z}^{(l)} &= 
Z_{a}^{(l)}\, 
\widehat{T}_{a}(u)\, Z_{a}^{(l)} \,.
\label{Z2leftmonod}
\end{align}
Evaluating ${\cal Z}^{(l)}\, t(u,p)\, {\cal Z}^{(l)}$ using the definition (\ref{transfer}) of the 
transfer matrix together with (\ref{Z2leftmonod}) and  (\ref{Z2leftK}), we arrive at the 
result (\ref{Z2leftprop}).
\end{proof}

\subsubsection{Action of the ``left'' $Z_{2}$ on the QG generators}

The gauge-transformed R-matrix (\ref{gaugeR}) for $D_{n}^{(1)}$ with $p<n$ obeys
\be
Z^{(l)}_{1}\, \tilde{R}_{12}(u,p)\, Z^{(l)}_{1} = Z^{(l)}_{2}\, 
\tilde{R}_{12}(u,p)\, Z^{(l)}_{2}\,,
\label{Z2leftRgauge}
\ee
in view of the property (\ref{Z2leftR}) and the fact that $\left[ 
Z^{(l)}\,, B(u,p) \right] = 0$ for $p<n$.
Hence, the gauge-transformed monodromy matrices (\ref{monodromygauge}) transform as 
follows
\begin{align}
{\cal Z}^{(l)}\, \tilde{T}_{a}(u,p)\, {\cal Z}^{(l)} &= Z_{a}^{(l)}\, 
\tilde{T}_{a}(u,p)\, Z_{a}^{(l)} \,, \non \\
{\cal Z}^{(l)}\, \widehat{\tilde{T}}_{a}(u,p)\, {\cal Z}^{(l)} &= 
Z_{a}^{(l)}\, 
\widehat{\tilde{T}}_{a}(u,p)\, Z_{a}^{(l)} \,.
\label{Z2leftmonodgauge}
\end{align}
Taking asymptotic limits of this result gives the following 
proposition:

\begin{prop}
For the case $D_{n}^{(1)}$ with $p<n$, the asymptotic 
gauge-transformed monodromy matrices $\tilde{T}^{\pm}_{a}(p)$ 
transform under the ``left'' $Z_{2}$ as   
\be
{\cal Z}^{(l)}\, \tilde{T}^{\pm}_{a}(p)\, {\cal Z}^{(l)} = 
Z_{a}^{(l)}\, \tilde{T}^{\pm}_{a}(p)\, Z_{a}^{(l)} \,.
\label{Z2leftmonotilde}
\ee 
\end{prop}

We can read off from this result how the coproducts of 
the ``left'' QG generators transform under this $Z_{2}$ symmetry. 
In particular, we observe that
\begin{align}
Z_{a}^{(l)}\, H^{(l)}_{j}\, Z_{a}^{(l)}  &= \left\{ 
\begin{array}{rl}
    H^{(l)}_{j} & \mbox{ for }  j = 1, \ldots, n-p-1\,, \\
   -H^{(l)}_{n-p} & \mbox{ for }  j= n-p 
   \end{array} \right. \,, \non \\
   Z_{a}^{(l)}\, E^{\pm (l)}_{j}\, Z_{a}^{(l)} &= \left\{   
\begin{array}{ll}
     E^{\pm (l)}_{j} & \mbox{ for }  j = 1, \ldots, n-p-2\,, \\
     E^{\pm (l)}_{n-p} & \mbox{ for }  j= n-p-1 
    \end{array} \right. \,.
\end{align}
Hence, this transformation maps complex representations of $U_{q}(D_{n-p})$ to 
their conjugates.

\section{Degeneracies of the transfer matrix}\label{sec:degen}

The symmetries identified above can be used to understand the 
degeneracies in the spectrum of the transfer matrix. Most importantly, 
the QG symmetries of the transfer matrix (\ref{prop1}), 
summarized in Table \ref{table:symmetries}, are
directly manifested in the degeneracies of the spectrum.  Indeed, for
generic values of the anisotropy parameter $\eta$, the $N$-site Hilbert space 
${\cal V}^{\otimes N}$ can be decomposed
into a direct sum of irreducible representations of the corresponding
classical group, whose dimensions are generally equal to the
degeneracies of the eigenvalues.

For the cases $A_{2n-1}^{(2)}$, $B_{n}^{(1)}$ and $D_{n}^{(1)}$, 
the transfer matrix has an additional ``right'' $Z_{2}$ symmetry (\ref{Z2rightprop}) 
that maps complex representations of $U_{q}(D_{p})$ to their conjugates.
Moreover, for the case $D_{n}^{(1)}$, the transfer matrix also has 
a  ``left'' $Z_{2}$ symmetry (\ref{Z2leftprop}) that maps complex
representations of $U_{q}(D_{n-p})$ to their conjugates. Consequently,
the degeneracies of eigenvalues corresponding to complex 
representations are larger than expected from the decomposition of the Hilbert space. 

For the cases $C_{n}^{(1)}$ and $D_{n}^{(1)}$ with $n$ even and 
$p=\frac{n}{2}$, the transfer matrix has a self-duality symmetry (\ref{selfdual})
that maps the representations $(1, {\bf R})$ and $({\bf R}, 1)$ into 
each other, and therefore those states are degenerate.  
If $\gamma_{0}=-1$, then there is a bonus symmetry (\ref{propbonus}), (\ref{calDsymmetry}) that 
leads to additional degeneracies.

For the cases $C_{n}^{(1)}$ and $D_{n}^{(1)}$ with $n$ odd and 
$p=\frac{n\pm 1}{2}$, we also observe some higher degeneracies, 
which presumably can also be attributed to some discrete symmetries
that remain to be elucidated.

We now consider examples of each of these cases.

\subsection{$A_{2n}^{(2)}$}

For $A_{2n}^{(2)}$ and generic values of $\eta$, the degeneracies of the transfer matrix exactly 
match with the predictions from the decomposition of the Hilbert 
space based on the QG symmetry. That is, in contrast with the other cases considered below, we do not 
find any higher degeneracies. As an example, let us consider the 
case $n=5$ and $N=2$ (two sites). By direct diagonalization 
of the transfer matrix $t(u, p)$ for generic numerical values of $u$ and 
$\eta$, we find that the degeneracies are as follows:
\begin{align}
& p=0:  &  &\{1, 55, 65 \}\non \\
& p=1:  &  &\{1, 1, 3, 18, 18, 36, 44 \}\non \\
& p=2:  &  &\{1, 1, 5, 10, 21, 27, 28, 28 \}\non \\
& p=3:  &  &\{1, 1, 10, 14, 14, 21, 30, 30 \}\non \\
& p=4:  &  &\{1, 1, 3, 5, 24, 24, 27, 36 \}\non \\
& p=5:  &  &\{1, 1, 10, 10, 44, 55 \} \,.
\label{numericaldegenA2n2}
\end{align}
In other words, for $p=0$, one eigenvalue is repeated 65 times, 
another eigenvalue is repeated 55 times, and another eigenvalue 
appears only once; and similarly for other values of $p$.

On the other hand, according to Table \ref{table:symmetries}, the 
symmetry for $A_{2n}^{(2)}$ with $n=5$ 
is $U_{q}(B_{5-p}) \otimes U_{q}(C_{p})$, and the representation at each site
is ${\cal V} =
(11-2p,1) \oplus (1,2p)$. For generic values of $\eta$, the 
QG representations are the same as for the corresponding classical groups.
Performing the tensor-product decompositions here and below using LieART 
\cite{Feger:2012bs}, we obtain\footnote{We recall that $A_{1} = B_{1} 
= C_{1}$, while the $D_{n}$ series starts with $n=2$.}
\begin{align}
& p=0: B_{5} & (\mathbf{11})^{\otimes 2} &= \mathbf{1} \oplus 
\mathbf{55} \oplus \mathbf{65} \non \\
& p=1: B_{4}\otimes C_{1} & ((\mathbf{9},\mathbf{1}) \oplus 
(\mathbf{1},\mathbf{2}))^{\otimes 2} &= 
2(\mathbf{1}, \mathbf{1}) \oplus (\mathbf{1}, \mathbf{3}) \oplus 
2(\mathbf{9}, \mathbf{2}) \oplus (\mathbf{36}, \mathbf{1}) \oplus 
(\mathbf{44}, \mathbf{1})\non \\
& p=2: B_{3}\otimes C_{2} & ((\mathbf{7},\mathbf{1}) \oplus 
(\mathbf{1},\mathbf{4}))^{\otimes 
2} &= 2(\mathbf{1}, \mathbf{1}) \oplus (\mathbf{1}, \mathbf{5}) 
\oplus 2 (\mathbf{7}, \mathbf{4}) \oplus (\mathbf{1}, \mathbf{10}) 
\oplus (\mathbf{21}, \mathbf{1}) \oplus (\mathbf{27}, \mathbf{1})  \non \\
& p=3: B_{2}\otimes C_{3} & ((\mathbf{5},\mathbf{1}) \oplus 
(\mathbf{1},\mathbf{6}))^{\otimes 
2} &= 2(\mathbf{1}, \mathbf{1}) \oplus 2(\mathbf{5}, \mathbf{6}) 
\oplus (\mathbf{10}, \mathbf{1}) \oplus (\mathbf{1}, \mathbf{14}) 
\oplus (\mathbf{14}, \mathbf{1}) \oplus (\mathbf{1}, \mathbf{21}) \non \\
& p=4: B_{1}\otimes C_{4} & ((\mathbf{3},\mathbf{1}) \oplus 
(\mathbf{1},\mathbf{8}))^{\otimes 
2} &= 2(\mathbf{1}, \mathbf{1}) \oplus (\mathbf{3}, \mathbf{1}) 
\oplus (\mathbf{5}, \mathbf{1}) \oplus 2(\mathbf{3}, \mathbf{8}) 
\oplus (\mathbf{1}, \mathbf{27}) \oplus (\mathbf{1}, \mathbf{36}) \non \\
& p=5: C_{5} & (\mathbf{1} \oplus \mathbf{10})^{\otimes 2} &=  
2(\mathbf{1}) \oplus 2(\mathbf{10}) \oplus \mathbf{44} \oplus 
\mathbf{55} \,.
\label{LieARTA2n2}
\end{align}
Comparing the degeneracies (\ref{numericaldegenA2n2}) with the 
corresponding tensor-product decompositions (\ref{LieARTA2n2}), we 
see that they exactly match. We obtain similar results for other 
values of $n$ and $N$. The special cases $p=0$ and $p=n$ are discussed 
further in \cite{Ahmed:2017mqq}.

\subsection{$A_{2n-1}^{(2)}$}

For $A_{2n-1}^{(2)}$ and generic values of $\eta$, the degeneracies 
of the transfer matrix either 
match with the predictions from QG symmetry, or are 
larger due to the ``right'' $Z_{2}$ symmetry (\ref{Z2rightprop}). As an example, let us consider the 
case $n=5$ and $N=2$ (two sites). By direct diagonalization 
of the transfer matrix $t(u,p)$ for generic numerical values of $u$ and 
$\eta$, we find that the degeneracies are as follows:
\begin{align}
& p=0:  &  &\{1, 44, 55 \}\non \\
& p=2:  &  &\{1, 1, 6, 9, 14, 21, 24, 24 \}\non \\
& p=3:  &  &\{1, 1, 5, 10, 15, 20, 24, 24 \}\non \\
& p=4:  &  &\{1, 1, 3, 16, 16, 28, 35 \}\non \\
& p=5:  &  &\{1, 45, 54 \} \,.
\label{numericaldegenA2m1n2}
\end{align}
Note that we exclude the case $p=1$.

On the other hand, according to Table \ref{table:symmetries}, the 
symmetry for $A_{2n-1}^{(2)}$ with $n=5$ and $p \ne 1$
is $U_{q}(C_{5-p}) \otimes U_{q}(D_{p})$, and the representation at each site
is ${\cal V} = (10-2p,1) \oplus (1,2p)$. The tensor-product decompositions are as follows:
\begin{align}
& p=0: C_{5} & (\mathbf{10})^{\otimes 2} &= \mathbf{1} \oplus 
\mathbf{44} \oplus \mathbf{55} \non \\
& p=2: C_{3}\otimes D_{2} & ((\mathbf{6},\mathbf{1}) \oplus 
(\mathbf{1},\mathbf{4}))^{\otimes 2} &= 2(\mathbf{1}, \mathbf{1}) 
\oplus (\mathbf{1}, \mathbf{3}) 
\oplus (\mathbf{1}, \mathbf{\bar 3}) 
\oplus 2 (\mathbf{6}, \mathbf{4}) \oplus (\mathbf{1}, \mathbf{9}) 
\non \\
& & & \qquad
\oplus (\mathbf{14}, \mathbf{1}) \oplus (\mathbf{21}, \mathbf{1})  \non \\
& p=3: C_{2}\otimes D_{3} & ((\mathbf{4},\mathbf{1}) \oplus 
(\mathbf{1},\mathbf{6}))^{\otimes 
2} &= 2(\mathbf{1}, \mathbf{1}) \oplus (\mathbf{5}, \mathbf{1}) 
\oplus 2(\mathbf{4}, \mathbf{6}) \oplus (\mathbf{10}, \mathbf{1}) 
\non \\
& & & \qquad
\oplus (\mathbf{1}, \mathbf{15}) \oplus (\mathbf{1}, \mathbf{20'}) \non \\
& p=4: C_{1}\otimes D_{4} & ((\mathbf{2},\mathbf{1}) \oplus 
(\mathbf{1},\mathbf{8}_{v}))^{\otimes 
2} &= 2(\mathbf{1}, \mathbf{1}) \oplus (\mathbf{3}, \mathbf{1}) 
\oplus 2(\mathbf{2}, \mathbf{8}_{v}) \oplus (\mathbf{1}, \mathbf{28}) 
\oplus (\mathbf{1}, \mathbf{35}_{v}) \non \\
& p=5: D_{5} & (\mathbf{10})^{\otimes 2} &=  
\mathbf{1} \oplus \mathbf{45} \oplus \mathbf{54} \,.
\label{LieARTA2m1n2}
\end{align}
Comparing the degeneracies (\ref{numericaldegenA2m1n2}) with the 
corresponding tensor-product decompositions (\ref{LieARTA2m1n2}), we 
see that they match, except for $p=2$. For the latter case, the 
degeneracies are larger, due to the ``right'' $Z_{2}$ symmetry 
mapping complex representations of $D_{p}$ to their 
conjugates (here, the $\mathbf{3}$ and $\mathbf{\bar 3}$).
We obtain similar results for other 
values of $n$ and $N$. The special cases $p=0$ and $p=n$ are discussed 
further in \cite{Nepomechie:2017hgw}.

\subsection{$B_{n}^{(1)}$}

For $B_{n}^{(1)}$ and generic values of $\eta$, the degeneracies 
of the transfer matrix also either 
match with the predictions from QG symmetry, or are 
larger due to the ``right'' $Z_{2}$ symmetry (\ref{Z2rightprop}). As an example, let us consider the 
case $n=5$ and $N=2$ (two sites). By direct diagonalization 
of the transfer matrix $t(u,p)$ for generic numerical values of $u$ and 
$\eta$, we find that the degeneracies are as follows:
\begin{align}
& p=0:  &  &\{1, 55, 65 \}\non \\
& p=2:  &  &\{1, 1, 6, 9, 21, 27, 28, 28 \}\non \\
& p=3:  &  &\{1, 1, 10, 14, 15, 20, 30, 30 \}\non \\
& p=4:  &  &\{1, 1, 3, 5, 24, 24, 28, 35 \}\non \\
& p=5:  &  &\{1, 1, 10, 10, 45, 54 \} \,.
\label{numericaldegenBn1}
\end{align}
Note that we again exclude the case $p=1$.

On the other hand, according to Table \ref{table:symmetries}, the 
symmetry for $B_{n}^{(1)}$ with $n=5$ and $p \ne 1$
is $U_{q}(B_{5-p}) \otimes U_{q}(D_{p})$, and the representation at each site
is  ${\cal V} = (11-2p,1) \oplus (1,2p)$. The tensor-product decompositions are as follows:
\begin{align}
& p=0: B_{5} & (\mathbf{11})^{\otimes 2} &= \mathbf{1} \oplus 
\mathbf{55} \oplus \mathbf{65} \non \\
& p=2: B_{3}\otimes D_{2} & ((\mathbf{7},\mathbf{1}) \oplus 
(\mathbf{1},\mathbf{4}))^{\otimes 2} &= 2(\mathbf{1}, \mathbf{1}) 
\oplus (\mathbf{1}, \mathbf{3}) 
\oplus (\mathbf{1}, \mathbf{\bar 3}) 
\oplus (\mathbf{1}, \mathbf{9}) \oplus 2(\mathbf{7}, \mathbf{4}) 
\non \\
& & & \qquad
\oplus (\mathbf{21}, \mathbf{1}) \oplus (\mathbf{27}, \mathbf{1})  \non \\
& p=3: B_{2}\otimes D_{3} & ((\mathbf{5},\mathbf{1}) \oplus 
(\mathbf{1},\mathbf{6}))^{\otimes 
2} &= 2(\mathbf{1}, \mathbf{1}) \oplus 2(\mathbf{5}, \mathbf{6}) 
\oplus (\mathbf{10}, \mathbf{1}) \oplus (\mathbf{14}, \mathbf{1}) 
\non \\
& & & \qquad
\oplus (\mathbf{1}, \mathbf{15}) \oplus (\mathbf{1}, \mathbf{20'}) \non \\
& p=4: B_{1}\otimes D_{4} & ((\mathbf{3},\mathbf{1}) \oplus 
(\mathbf{1},\mathbf{8}_{v}))^{\otimes 
2} &= 2(\mathbf{1}, \mathbf{1}) \oplus (\mathbf{3}, \mathbf{1}) 
\oplus (\mathbf{5}, \mathbf{1}) \oplus 2(\mathbf{3},\mathbf{8}_{v}) \oplus (\mathbf{1}, \mathbf{28}) 
\oplus (\mathbf{1}, \mathbf{35}_{v}) \non \\
& p=5: D_{5} & (\mathbf{1} \oplus \mathbf{10})^{\otimes 2} &=  
2(\mathbf{1}) \oplus 2(\mathbf{10}) \oplus \mathbf{45} \oplus \mathbf{54} \,.
\label{LieARTBn1}
\end{align}
Comparing the degeneracies (\ref{numericaldegenBn1}) with the 
corresponding tensor-product decompositions (\ref{LieARTBn1}), we 
see that they match, except for $p=2$. For the latter case, the 
degeneracies are larger, due to the 
``right'' $Z_{2}$ symmetry mapping complex representations of $D_{p}$ to their 
conjugates (here, the $\mathbf{3}$ and $\mathbf{\bar 3}$).
We obtain similar results for other 
values of $n$ and $N$.

\subsection{$C_{n}^{(1)}$}

For $C_{n}^{(1)}$ and generic values of $\eta$, the degeneracies 
of the transfer matrix match with the predictions from QG symmetry, 
except when $n$ is even and $p=\frac{n}{2}$ (in which case there is a 
self-duality symmetry (\ref{selfdual})) or when $n$ is odd and 
$p=\frac{n\pm 1}{2}$. Moreover, 
the spectrum exhibits a $p \rightarrow n-p$ duality symmetry.

\subsubsection{Example 1: even $n$}\label{sec:Cneven}

As a first example, let us consider the 
case $n=4$ and $N=2$ (two sites). By direct diagonalization 
of the transfer matrix $t(u,p)$ for generic numerical values of $u$ and 
$\eta$, we find that the degeneracies are as follows: 
\begin{align}
& p=0:  &  &\{1, 27, 36 \}\non \\
& p=1:  &  &\{1, 1, 3, 12, 12, 14, 21 \}\non \\
& p=2:  &  &\begin{cases}
\{1, 1, 10, 16, 16, 20 \} & \mbox{for } \gamma_{0}=+1\\
\{2, 10, 20, 32 \} & \mbox{for } \gamma_{0}=-1
\end{cases}
\non \\
& p=3:  &  &\{1, 1, 3, 12, 12, 14, 21 \}\non \\
& p=4:  &  &\{1, 27, 36 \} \,.
\label{numericaldegenCn1x}
\end{align}
The fact that the degeneracies are the same for $p$ and $n-p$ is a 
consequence of the duality symmetry (\ref{duality}), (\ref{duality2}).

On the other hand, according to Table \ref{table:symmetries}, the 
symmetry for $C_{n}^{(1)}$ with $n=4$ 
is $U_{q}(C_{4-p}) \otimes U_{q}(C_{p})$, and the representation at each site
is ${\cal V} = (8-2p,1) \oplus (1,2p)$. The tensor-product decompositions are as follows:
\begin{align}
& p=0: C_{4} & (\mathbf{8})^{\otimes 2} &= \mathbf{1} \oplus 
\mathbf{27} \oplus \mathbf{36} \non \\
& p=1: C_{3}\otimes C_{1} & ((\mathbf{6},\mathbf{1}) \oplus 
(\mathbf{1},\mathbf{2}))^{\otimes 2} &= 2(\mathbf{1}, \mathbf{1}) 
\oplus (\mathbf{1}, \mathbf{3}) 
\oplus 2 (\mathbf{6}, \mathbf{2}) \oplus (\mathbf{14}, \mathbf{1}) 
\oplus (\mathbf{21}, \mathbf{1})  \non \\
& p=2: C_{2}\otimes C_{2} & ((\mathbf{4},\mathbf{1}) \oplus 
(\mathbf{1},\mathbf{4}))^{\otimes 2} &= 2(\mathbf{1}, \mathbf{1}) 
\oplus (\mathbf{5}, \mathbf{1}) \oplus (\mathbf{1}, \mathbf{5})  
\oplus 2 (\mathbf{4}, \mathbf{4}) \oplus (\mathbf{10}, \mathbf{1}) 
\oplus (\mathbf{1}, \mathbf{10})  \,.
\label{LieARTCn1x}
\end{align}
There is no need to display the tensor-product decompositions for 
$p>2$ due to the symmetry $p \rightarrow n-p$.

Comparing the degeneracies (\ref{numericaldegenCn1x}) with the 
corresponding tensor-product decompositions (\ref{LieARTCn1x}), we 
see that they match, except for $p=2$. For the latter case, the 
degeneracies are larger, due to the self-duality symmetry (\ref{selfdual}) 
for even $n$ and $p=\frac{n}{2}$, which here maps $(\mathbf{1}, 
\mathbf{5})$ to $(\mathbf{5}, \mathbf{1})$ (resulting 
in a 10-fold degeneracy), and also maps $(\mathbf{1}, 
\mathbf{10})$ to $(\mathbf{10}, \mathbf{1})$ (resulting 
in a 20-fold degeneracy). If $\gamma_{0}=-1$, then the 
bonus symmetry (\ref{propbonus}), (\ref{calDsymmetry}) implies that 
the two $ (\mathbf{4}, \mathbf{4})$ are degenerate (giving rise to a 
32-fold degeneracy), as well as the 
two $(\mathbf{1}, \mathbf{1})$ (resulting in a 2-fold degeneracy).

\subsubsection{Example 2: odd $n$}\label{sec:Cnodd}

As a second example, let us consider the 
case $n=5$ and $N=2$ (two sites). By direct diagonalization 
of the transfer matrix $t(u,p)$ for generic numerical values of $u$ and 
$\eta$, we find that the degeneracies are as follows: 
\begin{align}
& p=0:  &  &\{1, 44, 55 \}\non \\
& p=1:  &  &\{1, 1, 3, 16, 16, 27, 36\} \non \\
& p=2:  &  &\{1, 1, 5, 21, 34, 38 \}\non \\
& p=3:  &  &\{1, 1, 5, 21, 34, 38 \}\non \\
& p=4:  &  &\{1, 1, 3, 16, 16, 27, 36\} \non \\
& p=5:  &  &\{1, 44, 55 \} \,.
\label{numericaldegenCn1}
\end{align}
We see again that the degeneracies are the same for $p$ and $n-p$, as a 
consequence of the duality symmetry (\ref{duality}), (\ref{duality2}).

On the other hand, according to Table \ref{table:symmetries}, the 
symmetry for $C_{n}^{(1)}$ with $n=5$ 
is $U_{q}(C_{5-p}) \otimes U_{q}(C_{p})$, and the representation at each site
is ${\cal V} = (10-2p,1) \oplus (1,2p)$. The tensor-product decompositions are as follows:
\begin{align}
& p=0: C_{5} & (\mathbf{10})^{\otimes 2} &= \mathbf{1} \oplus 
\mathbf{44} \oplus \mathbf{55} \non \\
& p=1: C_{4}\otimes C_{1} & ((\mathbf{8},\mathbf{1}) \oplus 
(\mathbf{1},\mathbf{2}))^{\otimes 2} &= 2(\mathbf{1}, \mathbf{1}) 
\oplus (\mathbf{1}, \mathbf{3}) 
\oplus 2 (\mathbf{8}, \mathbf{2}) \oplus (\mathbf{27}, \mathbf{1}) 
\oplus (\mathbf{36}, \mathbf{1})  \non \\
& p=2: C_{3}\otimes C_{2} & ((\mathbf{6},\mathbf{1}) \oplus 
(\mathbf{1},\mathbf{4}))^{\otimes 2} &= 2(\mathbf{1}, \mathbf{1}) 
\oplus (\mathbf{1}, \mathbf{5}) 
\oplus 2 (\mathbf{6}, \mathbf{4}) \oplus (\mathbf{1}, \mathbf{10}) 
\oplus (\mathbf{14}, \mathbf{1}) \oplus (\mathbf{21}, \mathbf{1})  \,.
\label{LieARTCn1}
\end{align}
Again, there is no need to display the tensor-product decompositions for 
$p>2$ due to the symmetry $p \rightarrow n-p$.

Comparing the degeneracies (\ref{numericaldegenCn1}) with the 
corresponding tensor-product decompositions (\ref{LieARTCn1}), we 
see that they match, except for $p=2$. For the latter case, the 
degeneracies are larger: the $(\mathbf{1}, \mathbf{10})$ and one  
$(\mathbf{6}, \mathbf{4})$ are degenerate (resulting in a 34-fold 
degeneracy); and the $(\mathbf{14}, \mathbf{1})$  and the other 
$(\mathbf{6}, \mathbf{4})$ are degenerate (resulting in a 38-fold 
degeneracy). We expect that such degeneracies for odd $n$ and 
$p=\frac{n\pm 1}{2}$ can be attributed to some discrete symmetries, 
which remain to be elucidated.

\subsection{$D_{n}^{(1)}$}

For $D_{n}^{(1)}$ and generic values of $\eta$, the degeneracies 
of the transfer matrix match with the predictions from QG symmetry, 
except for the following exceptions:
when $n$ is even and $p=\frac{n}{2}$ (in which case there is a 
self-duality symmetry (\ref{selfdual})); when $n$ is odd and 
$p=\frac{n\pm 1}{2}$; and when there are additional 
degeneracies due to the
``right'' and ``left'' $Z_{2}$ symmetries (\ref{Z2rightprop}), 
(\ref{Z2leftprop}).  Moreover, 
the spectrum exhibits a $p \rightarrow n-p$ duality symmetry.

\subsubsection{Example 1: even $n$}\label{sec:Dneven}

As a first example, let us consider the 
case $n=6$ and $N=2$ (two sites). By direct diagonalization 
of the transfer matrix $t(u,p)$ for generic numerical values of $u$ and 
$\eta$, we find that the degeneracies are as follows: 
\begin{align}
& p=0:  &  &\{1, 66, 77 \}\non \\
& p=2:  &  &\{1, 1, 6, 9, 28, 32, 32, 35 \}\non \\
& p=3:  &  &\begin{cases}
\{1, 1, 30, 36, 36, 40 \}  & \mbox{for } \gamma_{0}=+1\\
\{2, 30, 40, 72 \}  & \mbox{for } \gamma_{0}=-1
\end{cases}
\non \\
& p=4:  &  &\{1, 1, 6, 9, 28, 32, 32, 35 \} \non\\
& p=6:  &  &\{1, 66, 77 \} \,.
\label{numericaldegenDn1x}
\end{align}
Note that we exclude the cases $p=1$ and $p=n-1$.
The fact that the degeneracies are the same for $p$ and $n-p$ is a 
consequence of the duality symmetry (\ref{duality}), (\ref{duality2}).

On the other hand, according to Table \ref{table:symmetries}, the 
symmetry for $D_{n}^{(1)}$ with $n=6$ and $p \ne 1, n-1$
is $U_{q}(D_{6-p}) \otimes U_{q}(D_{p})$, and the representation at each site
is ${\cal V} = (12-2p,1) \oplus (1,2p)$. The tensor-product decompositions are as follows:
\begin{align}
& p=0: D_{6} & (\mathbf{12})^{\otimes 2} &= \mathbf{1} \oplus 
\mathbf{66} \oplus \mathbf{77} \non \\
& p=2: D_{4}\otimes D_{2} & ((\mathbf{8}_{v},\mathbf{1}) \oplus 
(\mathbf{1},\mathbf{4}))^{\otimes 2} &= 2(\mathbf{1}, \mathbf{1}) 
\oplus (\mathbf{1}, \mathbf{3})  \oplus (\mathbf{1}, \mathbf{\bar 
3}) \oplus (\mathbf{1}, \mathbf{9})  \oplus 2 (\mathbf{8}_{v}, 
\mathbf{4}) \non\\
& & & \qquad
\oplus (\mathbf{28}, \mathbf{1}) \oplus (\mathbf{35}_{v}, \mathbf{1}) \non\\
& p=3: D_{3}\otimes D_{3} & ((\mathbf{6},\mathbf{1}) \oplus 
(\mathbf{1},\mathbf{6}))^{\otimes 2} &= 2(\mathbf{1}, \mathbf{1}) 
\oplus 2(\mathbf{6}, \mathbf{6})  \oplus (\mathbf{15}, \mathbf{1}) 
\oplus (\mathbf{1}, \mathbf{15})  \non\\
& & & \qquad
\oplus (\mathbf{20'}, \mathbf{1}) \oplus (\mathbf{1}, \mathbf{20'}) \,.
\label{LieARTDn1x}
\end{align}
There is no need to display the tensor-product decompositions for 
$p>3$ due to the symmetry $p \rightarrow n-p$.

Comparing the degeneracies (\ref{numericaldegenDn1x}) with the 
corresponding tensor-product decompositions (\ref{LieARTDn1x}), we 
see that they match for $p=0$. For $p=2$, the 
degeneracies are larger due to the the ``right'' $Z_{2}$ symmetry 
(\ref{Z2rightprop}), which maps $(\mathbf{1}, \mathbf{3})$ to $(\mathbf{1}, \mathbf{\bar 
3})$, and results in a 6-fold degeneracy. 

For $p=3$, the
degeneracies are larger due to the self-duality symmetry (\ref{selfdual}) 
for even $n$ and $p=\frac{n}{2}$, which maps $(\mathbf{1}, 
\mathbf{15})$ to $(\mathbf{15}, \mathbf{1})$ (resulting in a 30-fold 
degeneracy), and also maps  $(\mathbf{1}, 
\mathbf{20'})$ to $(\mathbf{20'}, \mathbf{1})$ (resulting in a 40-fold 
degeneracy).  If $\gamma_{0}=-1$, then the 
bonus symmetry (\ref{propbonus}), (\ref{calDsymmetry}) implies that 
the two $ (\mathbf{6}, \mathbf{6})$ are degenerate (giving rise to a 
72-fold degeneracy), as well as the 
two $(\mathbf{1}, \mathbf{1})$ (resulting in a 2-fold degeneracy).

\subsubsection{Example 2: odd $n$}\label{sec:Dnodd}

As a second example, let us consider the 
case $n=5$ and $N=2$ (two sites). By direct diagonalization 
of the transfer matrix $t(u,p)$ for generic numerical values of $u$ and 
$\eta$, we find that the degeneracies are as follows: 
\begin{align}
& p=0:  &  &\{1, 45, 54 \}\non \\
& p=2:  &  &\{1, 1, 6, 20, 33, 39 \}\non \\
& p=3:  &  &\{1, 1, 6, 20, 33, 39 \}\non \\
& p=5:  &  &\{1, 45, 54 \} \,.
\label{numericaldegenDn1}
\end{align}
We again exclude the cases $p=1, n-1$, and observe that 
the degeneracies are the same for $p$ and $n-p$, as a 
consequence of the duality symmetry (\ref{duality}), (\ref{duality2}).

On the other hand, according to Table \ref{table:symmetries}, the 
symmetry for $D_{n}^{(1)}$ with $n=5$ and $p \ne 1, n-1$
is $U_{q}(D_{5-p}) \otimes U_{q}(D_{p})$, and the representation at each site
is ${\cal V} = (10-2p,1) \oplus (1,2p)$. The tensor-product decompositions are as follows:
\begin{align}
& p=0: D_{5} & (\mathbf{10})^{\otimes 2} &= \mathbf{1} \oplus 
\mathbf{45} \oplus \mathbf{54} \non \\
& p=2: D_{3}\otimes D_{2} & ((\mathbf{6},\mathbf{1}) \oplus 
(\mathbf{1},\mathbf{4}))^{\otimes 2} &= 2(\mathbf{1}, \mathbf{1}) 
\oplus (\mathbf{1}, \mathbf{3}) \oplus (\mathbf{1}, \mathbf{\bar 3}) 
\oplus 2 (\mathbf{6}, \mathbf{4}) \oplus (\mathbf{1}, \mathbf{9}) \non\\
& & & \qquad
\oplus (\mathbf{15}, \mathbf{1}) \oplus (\mathbf{20'}, \mathbf{1})  \,.
\label{LieARTDn1}
\end{align}
Again, there is no need to display the tensor-product decompositions for 
$p>2$ due to the symmetry $p \rightarrow n-p$.

Comparing the degeneracies (\ref{numericaldegenDn1}) with the 
corresponding tensor-product decompositions (\ref{LieARTDn1}), we 
see that they match for $p=0$. For $p=2$, the 6-fold degeneracy
is due to ``right'' $Z_{2}$ symmetry, which maps $(\mathbf{1}, \mathbf{3})$ 
to $(\mathbf{1}, \mathbf{\bar 3})$. Moreover, the $(\mathbf{1}, 
\mathbf{9})$ and one  
$(\mathbf{6}, \mathbf{4})$ are degenerate (resulting in a 33-fold 
degeneracy); and the $(\mathbf{15}, \mathbf{1})$  and the other 
$(\mathbf{6}, \mathbf{4})$ are degenerate (resulting in a 39-fold 
degeneracy). We expect that such degeneracies for odd $n$ and 
$p=\frac{n\pm 1}{2}$ can be attributed to some discrete symmetries, 
which remain to be elucidated.

\section{Outlook}\label{sec:outlook}

Several interesting problems remain to be addressed, some of which we 
list here.

We have noted the existence of a higher degeneracy of the transfer 
matrix that occurs for the cases $C_{n}^{(1)}$ and $D_{n}^{(1)}$ with $n$ odd and 
$p=\frac{n\pm 1}{2}$, see Sections \ref{sec:Cnodd} and 
\ref{sec:Dnodd}. These degeneracies are unusual, since they result 
from the ``mixing'' of representations of {\em unequal} dimensions,
such as the $(\mathbf{1}, \mathbf{10})$ and the $(\mathbf{6}, 
\mathbf{4})$ discussed in Section \ref{sec:Cnodd}.
In contrast, the self-duality and $Z_{2}$ symmetries that we 
identified imply degeneracies of representations of {\em equal} 
dimensions, namely,  $(1, {\bf R}) \leftrightarrow ({\bf R}, 1)$ and 
${\bf R} \leftrightarrow {\bf \bar R}$, respectively. It would be interesting to 
find some discrete symmetries that could account for these unusual degeneracies.

For the R-matrices that we have considered (\ref{affinealgebras}), the
$K$-matrices (\ref{KR}) do not exhaust the possible diagonal
K-matrices.  Indeed, a few additional diagonal solutions depending on
one boundary parameter are known \cite{Malara:2004bi}.  We expect that
the corresponding transfer matrices also have some QG symmetry;
however, we leave an investigation of those cases to the future.

We have not considered here the case of the $D^{(2)}_{n+1}$ R-matrix
\cite{Jimbo:1985ua}, because a corresponding set of K-matrices
depending on an integer $p=0, 1, \ldots, n$ is not yet known.  It
would be interesting to find such a set of K-matrices, since the
corresponding transfer matrices would presumably have the QG symmetry
$U_{q}(B_{n-p}) \otimes U_{q}(B_{p})$, as well as a $p \leftrightarrow
n-p$ duality symmetry, and a self-duality symmetry for even $n$ and
$p=\frac{n}{2}$.  So far, only the special cases $p=0$ and $p=n$ have
been investigated \cite{Nepomechie:2017hgw}, based on the
$D^{(2)}_{n+1}$ K-matrices found in \cite{Martins:2000xie}.
Interestingly, these K-matrices are not diagonal, but only
block-diagonal.

Since the transfer matrix $t(u,p)$ is integrable 
(\ref{commutativity}), its eigenvalues and eigenvectors can be 
determined by Bethe ansatz. We expect that, for general values of 
$p$, the $(g^{(l)}, g^{(r)})$ Dynkin labels of the Bethe states
can be related to the numbers of Bethe roots of 
each type, as was done for $p=0$ and $p=n$ in \cite{Ahmed:2017mqq, Nepomechie:2017hgw}.
It would be interesting to understand the dependence of the Bethe 
equations on $p$; and, for the cases $C_{n}^{(1)}$ and $D_{n}^{(1)}$, 
to see how the $p \leftrightarrow n-p$ duality symmetry is manifested 
in these Bethe equations.

\section*{Acknowledgments}
We thank J. F. Gomes and M. Jimbo for helpful correspondence.
RN was supported in part by a Cooper fellowship, and AR was supported by the S\~ao Paulo
Research Foundation FAPESP under the process  \# 2017/03072-3 and \# 
2015/00025-9. AR thanks the University of Miami for its warm 
hospitality.

\appendix

\section{R-matrices}\label{sec:Rmatices}

The R-matrices are given by 
\begin{align}
R(u) &=  c(u)\sum_{\alpha\neq\alpha'}e_{\alpha\alpha}\otimes e_{\alpha\alpha}
+ b(u)\sum_{\alpha\neq\beta,\beta'}e_{\alpha\alpha}\otimes e_{\beta\beta}\nonumber\\
& +\left(\, e(u)\sum_{\alpha<\beta,\alpha\neq\beta'}
+\,\bar{e}(u)\sum_{\alpha>\beta,\alpha\neq\beta'}\right)e_{\alpha\beta}\otimes e_{\beta\alpha}
+\sum_{\alpha,\beta}a_{\alpha\beta}(u)e_{\alpha\beta}\otimes 
e_{\alpha'\beta'} \,,
\label{Rmat}
\end{align}
where $e_{\alpha\beta}$ are the elementary $d \times d$ matrices, 
with $d$ given by (\ref{defd}). Moreover,
\be
\left. \begin{array}{lll}
      c(u)&=&2 \sinh(\frac{u}{2}-2\eta)\   \\
      b(u)&=&2 \sinh(\frac{u}{2})\ \\
      e(u)&=&-2 e^{-\frac{u}{2}} \sinh (2\eta)\  \\
      \end{array}\right \} \times
\left \{ \begin{array}{lll}
       \cosh(\frac{u}{2} - \kappa\eta) & \mbox{for} & 
       A_{2n}^{(2)}\,, A_{2n-1}^{(2)} \\
\sinh(\frac{u}{2} -\kappa\eta)  & \mbox{for} & B_{n}^{(1)}\,, 
C_{n}^{(1)}\,, D_{n}^{(1)}
       \end{array} \right. \,,
\label{Raa}
\ee
\be
\bar{e}(u)=e^u e(u)\,, \non
\ee

\begin{align}
a_{\alpha \beta}(u) &=
\left\{ \begin{array}{l}
        2\sinh(\frac{u}{2})\times\left\{ \begin{array}{l}
     \cosh(\frac{u}{2} -(\kappa-2)\eta)\quad \mbox{for} \quad 
     A_{2n}^{(2)}\,, A_{2n-1}^{(2)} \\
     \sinh(\frac{u}{2} -(\kappa-2)\eta)\quad \mbox{for} \quad B_n^{(1)} \,,
        C_n^{(1)} \,, D_n^{(1)}
                                 \end{array}\right.
              \qquad \alpha =\beta, \alpha \neq \alpha'\\
              \\
       b(u) + \left\{\begin{array}{ll}
          2 \sinh(2\eta)\sinh((2n-1)\eta) & \mbox{for} \quad B_n^{(1)} \\
         -2 \sinh(2\eta)\cosh((2n+1)\eta) & \mbox{for} \quad A_{2n}^{(2)}
	 \end{array}\right.
              \qquad \alpha=\beta, \alpha=\alpha'\\
               \\
       2\sinh (2\eta) e^{\mp\frac{u}{2}}\times\left\{ \begin{array}{ll}
                      \mp\epsilon_\alpha\epsilon_\beta
                    e^{(\pm \kappa 
		    +2(\bar{\alpha}-\bar{\beta}))\eta}\sinh(\frac{u}{2})\\
		    \quad -\delta_{\alpha \beta'} \cosh(\frac{u}{2}-\kappa\eta)
& \mbox{for} \quad A_{2n}^{(2)}\,, A_{2n-1}^{(2)} \\ \non\\
                      \epsilon_\alpha\epsilon_\beta
                    e^{(\pm \kappa 
		    +2(\bar{\alpha}-\bar{\beta}))\eta}\sinh(\frac{u}{2})\\
                      \quad  -\delta_{\alpha \beta'} \sinh(\frac{u}{2}-\kappa\eta)
& \mbox{for} \quad B_n^{(1)} \,, C_n^{(1)} \,, D_n^{(1)}\\
                                  \end{array}\right.
              \quad \alpha _>^< \beta
       \end{array} \right. 
       & \\
\label{Ra}
\end{align}
where 
\be
\kappa = \left\{ \begin{array}{l l}
2n  &   \mbox{ for  } A_{2n-1}^{(2)} \\
2n+1 &  \mbox{ for  } A_{2n}^{(2)}\\
2n-1 &  \mbox{ for  } B_{n}^{(1)}\\
2n+2 &  \mbox{ for  } C_{n}^{(1)}\\
2n-2 &  \mbox{ for  } D_{n}^{(1)}
\end{array} \right. \,,
\label{kappa}
\ee
\begin{align}
\epsilon_\alpha & =  \left \{\begin{array}{rll}
                        1 & \mbox{  for   } & 1\le\alpha \le n \\
                       -1 & \mbox{  for   } & n+1 \le \alpha \le 2n
                               \end{array}\right.  & \mbox{  for  }  
			       A_{2n-1}^{(2)}\,, C_n^{(1)}
			       \non\\
\epsilon_\alpha & =  1 & \quad \mbox{  for  } A_{2n}^{(2)}\,, B_n^{(1)}\,, D_n^{(1)} 
\label{epsilonalpha}
\end{align}
\begin{align}
\bar{\alpha} &=\left\{ \begin{array}{ll}
             \alpha-\half & 1\le\alpha \le n\\
             \alpha+\half & n+1\le\alpha\le 2n
             \qquad\mbox{for}\,\,A_{2n-1}^{(2)}\,,  C_n^{(1)}
                     \end{array} \right. \non \\
\bar{\alpha} &=\left\{ \begin{array}{ll}
              \alpha+\half & 1\le\alpha<\frac{d+1}{2}\\
                    \alpha & \alpha=\frac{d+1}{2}\\
              \alpha-\half & \frac{d+1}{2}<\alpha\le d
              \qquad\mbox{for}\,\, A_{2n}^{(2)}\,, B_n^{(1)}\,, D_n^{(1)}
                     \end{array} \right.
\label{bar}		     
\end{align}
\begin{align}
        \alpha' &= d+1-\alpha  \,, \non\\
\alpha,\beta &= 1 \,, \ldots \,, d \,.
\label{prime}
\end{align}

All but one of these R-matrices are the same as in 
\cite{Jimbo:1985ua}, up to the change of variables $x=e^{u}\,, 
k=e^{2\eta}$ and an overall factor. The one exception is the R-matrix 
for $A_{2n-1}^{(2)}$, which we obtain from the $C_{n}^{(1)}$ R-matrix in \cite{Jimbo:1985ua} by
replacing $\xi=k^{2n+2}$ by $\xi=-k^{2n}$; i.e. by changing $\xi
\mapsto -\xi k^{-2}$.  It is the same as the $A_{2n-1}^{(2)}$ R-matrix
in the appendix of \cite{Kuniba:1991yd} up to some redefinitions of the anisotropy and 
spectral parameters, and an overall factor. This $A_{2n-1}^{(2)}$ R-matrix
was used in \cite{Artz:1995bm, Li:2006mv, Nepomechie:2017hgw}.

\section{$U_{q}(g^{(l)}) \otimes U_{q}(g^{(r)})$ and 
$\tilde{T}^{\pm}(p)$}\label{sec:QG}

We show here that the asymptotic gauge-transformed monodromy matrix
$\tilde{T}^{\pm}(p)$ (\ref{tildeTpm}) can be expressed in terms of
coproducts of the generators of a QG of the form $U_{q}(g^{(l)})
\otimes U_{q}(g^{(r)})$, where $g^{(l)}$ and $g^{(r)}$ are
(non-affine) simple Lie algebras of type $B$, $C$ or $D$, with rank
$n-p$ and $p$, respectively.  Specifically, the pairs of algebras
$(g^{(l)}, g^{(r)})$ are given in Table \ref{table:algebras}, where
$\hat g$ is the affine Lie algebra in the list (\ref{affinealgebras})
that is associated to the R-matrix. 
The algebras $g^{(l)} \oplus g^{(r)}$
are in fact the subalgebras of $\hat g$ obtained by removing
the $p^{th}$ node from the (extended) Dynkin diagram of $\hat g$, 
which has $n+1$ nodes.  We emphasize
that the possible values of $p$ are $0, 1, \ldots, n$; it is
understood that the ``right'' algebra
$g^{(r)}$ is absent for $p=0$, while the ``left'' algebra $g^{(l)}$ is
absent for $p=n$.

\begin{table}[h]
\centering
\begin{tabular}{|c|c|}
\hline
$\hat g$ & $(g^{(l)}, g^{(r)})$ \\   
\hline
$A_{2n}^{(2)}$ & $(B_{n-p}, C_{p})$ \\
$A_{2n-1}^{(2)}$ & $(C_{n-p}, D_{p})\quad (p \ne 1)$ \\
$B_{n}^{(1)}$ & $(B_{n-p}, D_{p}) \quad (n>1, p \ne 1)$ \\
$C_{n}^{(1)}$ & $(C_{n-p}, C_{p})$ \\
$D_{n}^{(1)}$ & $(D_{n-p}, D_{p}) \quad (n>1, p \ne 1, 
n-1)$ \\
\hline
\end{tabular}
\caption{\small Pairs of Lie algebras $(g^{(l)}, g^{(r)})$ 
corresponding to the affine Lie algebras $\hat g$, 
where $p=0, 1, \ldots, n$.}
\label{table:algebras}
\end{table}

\subsection{Generators}

We denote the generators corresponding to the simple roots of $g^{(l)}$ and 
$g^{(r)}$ by
\be
H^{(l)}_{i}(p)\,, \quad E^{\pm\, (l)}_{i}(p) \,, \qquad i = 1, \ldots, 
n-p\,,  \non
\ee
and 
\be
H^{(r)}_{i}(p)\,, \quad E^{\pm\, (r)}_{i}(p) \,, \qquad i = 1, \ldots, p \,, \non
\ee
respectively. (To lighten the notation, we shall refrain from displaying the dependence 
of these generators on $p$ when there is no ambiguity in so doing.)
The ``left'' generators satisfy the commutation relations
\begin{align}
\left[H^{(l)}_i(p)\,, H^{(l)}_j(p) \right] &= 0 \,, \non\\
\left[H^{(l)}_i(p)\,, E^{\pm\, (l)}_j(p) \right] &= \pm \alpha_i^{(j)}E^{\pm\, (l)}_j(p) \,, \non\\
\left[E^{+\, (l)}_i(p)\,, E^{-\, (l)}_j(p) \right] &=\delta_{i,j}\sum_{k=1}^{n-p}\alpha_k^{(j)}H^{(l)}_k(p) \,,
\label{leftcommutators}
\end{align}
and the ``right'' generators similarly satisfy the commutation relations 
\begin{align}
\left[H^{(r)}_i(p)\,, H^{(r)}_j(p) \right] &= 0 \,, \non\\
\left[H^{(r)}_i(p)\,, E^{\pm\, (r)}_j(p) \right] &= \pm 
\alpha_i^{(j)}E^{\pm\, (r)}_j(p) \,, \non\\
\left[E^{+\, (r)}_i(p)\,, E^{-\, (r)}_j(p) \right] 
&=\delta_{i,j}\sum_{k=1}^{p}\alpha_k^{(j)}H^{(r)}_k(p) \,.
\label{rightcommutators}
\end{align}
Moreover, the ``left'' and  
``right'' generators commute with each other
\be
\left[ H^{(l)}_{i}(p)\,, E^{\pm\, (r)}_{j}(p) \right] = 
\left[ E^{\pm\, (l)}_{i}(p)\,,  H^{(r)}_{j}(p) \right] = 
\left[ E^{\pm\, (l)}_{i}(p)\,, E^{\pm\, (r)}_{j}(p) \right] = 
\left[ E^{\pm\, (l)}_{i}(p)\,, E^{\mp\, (r)}_{j}(p) \right] = 0 \,.
\label{leftrightcommutator}
\ee
The simple roots $\{\alpha^{(1)}, 
\ldots,  \alpha^{(m)} \}$ (where $m$ is either $n-p$ or $p$) in the orthogonal basis are given by
\begin{align}
\alpha^{(j)} &=   e_{j} - e_{j+1} \,, \qquad  j =  1,\ldots, m-1\,,   
\non \\
\alpha^{(m)} &=  \left\{
\begin{array}{ll}
    e_{m} & \mbox{ for  } B_{m} \\
   2e_{m} & \mbox{ for  } C_{m} \\
   e_{m-1} + e_{m} & \mbox{ for  } D_{m}
    \end{array} \right. \,,
\label{simpleroots}
\end{align}
where $e_{j}$ are the elementary $m$-dimensional basis vectors 
$(e_{j})_{i} = \delta_{i,j}$ (i.e., $e_{1}=(1, 0, 0, 
\ldots, 0)\,, e_{2} = (0, 1, 0, \ldots, 0)$, etc.). 

In terms of the $\hat g$ generators\footnote{Note that $e_{ij}$ 
are the elementary $d \times d$ matrices introduced below 
(\ref{permutation}), where $d$ is defined in (\ref{defd}).  We see from (\ref{leftgensvectorrep}) that the 
generators in (\ref{gensp0}) with $i=1, \ldots, n$ are in fact the 
generators of $g^{(l)}$ with $p=0$; and we see from (\ref{rightgensvectorrep}) that  
$E^{\pm}_{0}$ in (\ref{gensp0}) are the $n^{th}$ generators of 
$g^{(r)}$ with $p=n$.
}
\begin{align}
    H_{i} &= e_{i,i} - e_{d+1-i,d+1-i} \,, \qquad i = 1, \ldots, n \,, \non \\
    E^{+}_{i} &= e_{i,i+1} + e_{d-i, d+1-i} \,, \qquad i = 1, \ldots, n-1 \,, \non \\
    E^{+}_{n} &= \left\{ \begin{array}{ll}	
  e_{n,n+1} + e_{d-n, d+1-n} &  \mbox{ if  } g^{(l)} = B_{n-p} \mbox{ 
  i.e., for  } A^{(2)}_{2n}\,, B^{(1)}_{n} \\
          \sqrt{2}e_{n,n+1} &  \mbox{ if  } g^{(l)} = C_{n-p} 
	  \mbox{ i.e., for  } A^{(2)}_{2n-1}\,, C^{(1)}_{n} \\
        e_{n-1,n+1} + e_{n, n+2} & \mbox{ if  } g^{(l)} = D_{n-p} 
	\mbox{ i.e., for  } D^{(1)}_{n}
    \end{array} \right. \,, \non \\
    E^{+}_{0} &=   \left\{ \begin{array}{ll}	
          \sqrt{2}e_{d,1} & \mbox{ if  } g^{(r)} = C_{p} \mbox{ i.e., for  } A^{(2)}_{2n}\,, C^{(1)}_{n} \\
        e_{d-1,1} + e_{d, 2} & \mbox{ if  } g^{(r)} = D_{p} \mbox{ 
	i.e., for  }  
	A^{(2)}_{2n-1}\,,B^{(1)}_{n}\,, D^{(1)}_{n}
    \end{array} \right. \,, \non \\
    E^{-}_{i} &= (E^{+}_{i})^{t} \,, \qquad\qquad\qquad i = 0, 1, 
    \ldots, n \,,
    \label{gensp0}
\end{align}
the ``left and ``right'' generators are given by
\begin{align}
H_{i}^{(l)}(p) &= H_{p+i} \,,  \non \\
E^{\pm\, (l)}_{i}(p) &=  E^{\pm}_{p+i} \,, \qquad i = 1, \ldots, n-p \,,
\label{leftgensvectorrep}
\end{align}
and
\begin{align}
H_{i}^{(r)}(p) &= -H_{p+1-i} \,,  \non \\
E^{\pm\, (r)}_{i}(p) &=  E^{\pm}_{p-i} \,, \qquad i = 1, \ldots, p \,,
\label{rightgensvectorrep}
\end{align}
respectively. Indeed, one can check that the commutation relations (\ref{leftcommutators}) 
-  (\ref{leftrightcommutator}) are satisfied. Note that the ``broken''
generators $E^{\pm}_{p}$ in (\ref{gensp0}) do not belong to either 
the ``left'' (\ref{leftgensvectorrep}) or ``right''  
(\ref{rightgensvectorrep}) set of generators; indeed, dropping the 
$\hat g$ generators
$E^{\pm}_{p}$ corresponds to deleting the $p^{th}$ node from the 
(extended) Dynkin diagram of $\hat g$.

\subsection{Coproducts}

We now present the coproducts for the quantum groups $U_{q}(g^{(l)})$ and 
$U_{q}(g^{(r)})$.

\subsubsection{``Left'' generators}

The coproducts for the ``left'' generators are given by
\begin{align}
\Delta(H^{(l)}_{j}) &= H^{(l)}_{j} \otimes \id + \id \otimes H^{(l)}_{j} \,, \qquad\qquad\qquad j = 
1, \ldots,  n-p \,, \non \\
\Delta(E^{\pm\, (l)}_{j}) &= E^{\pm\, (l)}_{j} \otimes e^{(\eta + i \pi)
H^{(l)}_{j} - \eta H^{(l)}_{j+1}} + e^{-(\eta + i \pi)
H^{(l)}_{j} + \eta H^{(l)}_{j+1}} \otimes E^{\pm\, (l)}_{j} \,, \qquad j 
= 1, \ldots, n-p-1 \,, \non  \\
\Delta(E^{\pm\, (l)}_{n-p}) &= \left\{
\begin{array}{ll}
E^{\pm\, (l)}_{n-p} \otimes e^{(\eta + i \pi) 
H^{(l)}_{n-p}} & \\
\qquad + e^{-(\eta + i \pi) 
H^{(l)}_{n-p} } \otimes E^{\pm\, (l)}_{n-p} &  
\mbox{ if  } g^{(l)} = B_{n-p} \mbox{ 
  i.e., for  } A^{(2)}_{2n}\,, B^{(1)}_{n} \\ \\
E^{\pm\, (l)}_{n-p} \otimes e^{2 \eta H^{(l)}_{n-p}} + e^{-2 \eta 
H^{(l)}_{n-p}} \otimes E^{\pm\, (l)}_{n-p} &  \mbox{ if  } g^{(l)} = 
C_{n-p}  \mbox{ i.e., for  } A^{(2)}_{2n-1}\,, C^{(1)}_{n} \\ \\
E^{\pm\, (l)}_{n-p} \otimes e^{ \eta 
	H^{(l)}_{n-p-1}+(\eta + i \pi) H^{(l)}_{n-p}}  & \\
\qquad 
+  e^{ - \eta H^{(l)}_{n-p-1}-(\eta + i \pi) H^{(l)}_{n-p}} \otimes E^{\pm\, (l)}_{n-p}
& \mbox{ if  } g^{(l)} = D_{n-p}  \mbox{ i.e., for  } D^{(1)}_{n}
\end{array} \right. \,.
\label{coproductleft}
\end{align}

These coproducts satisfy
\be
\left[  \Delta(H^{(l)}_{i}) \,,   \Delta(E^{\pm\, (l)}_{j}) \right] = \pm \alpha_{i}^{(j)} 
\Delta(E^{\pm\, (l)}_{j})\,,  
\label{DeltaHDeltaE}
\ee
and 
\begin{align}
& \Omega_{ij}^{(l)}\Delta(E_i^{+\, (l)})\Delta(E_j^{-\, (l)})-\Delta(E_j^{-\, 
(l)})\Delta(E_i^{+\, (l)})\Omega_{ij}^{(l)} \\
&= \begin{cases}
\delta_{i,j}\frac{\sinh\left[2\eta \sum_{k=1}^{n-p}\alpha_{k}^{(j)}\Delta(H^{(l)}_k)\right]}
{\sinh(2\eta)} & \mbox{for } A_{2n}^{(2)}\,, B_{n}^{(1)}\,, 
D_{n}^{(1)}\\ \\
\delta_{i,j}(1+\delta_{i,n-p})\frac{\sinh\left[2\eta \sum_{k=1}^{n-p}\alpha_{k}^{(j)}\Delta(H^{(l)}_k)\right]}
{\sinh(2(1+\delta_{i,n-p})\eta)} & \mbox{for } A_{2n-1}^{(2)}\,, C_{n}^{(1)}
\end{cases} \,,
\label{DeltaEDelta}
\end{align} 
where $\Omega_{ij}^{(l)}$ is given by 
\be
\Omega_{ij}^{(l)}=
\begin{cases}
\begin{cases}
e^{i \pi H^{(l)}_{\text{max}(i,j)}}\otimes\id & |i-j|=1 \\
\quad \, \quad \id\otimes\id & \mbox{ otherwise }
\end{cases} \quad  \mbox{for } A_{2n}^{(2)}\,, B_{n}^{(1)} \\ \\
\begin{cases}
e^{i \pi H^{(l)}_{\text{max}(i,j)}}\otimes\id & |i-j|=1 \mbox{ and } 1 \le 
\text{min}(i,j) \le n-p-2 \\
\quad \, \quad \id\otimes\id & \mbox{ otherwise }
\end{cases} \quad  \mbox{for }  A_{2n-1}^{(2)}\,, C_{n}^{(1)}\,, D_{n}^{(1)}
\end{cases}
\,.
\label{omegal}
\ee

\subsubsection{``Right'' generators}

The coproducts for the ``right'' generators are given by
\begin{align}
\Delta(H^{(r)}_{j}) &= H^{(r)}_{j} \otimes \id + \id \otimes H^{(r)}_{j} \,, \qquad\qquad\qquad j = 
1, \ldots,  p \,, \non \\
\Delta(E^{\pm\, (r)}_{j}) &= E^{\pm\, (r)}_{j} \otimes e^{(\eta + i \pi)
H^{(r)}_{j} - \eta H^{(r)}_{j+1}} + e^{-(\eta + i \pi)
H^{(r)}_{j} + \eta H^{(r)}_{j+1}} \otimes E^{\pm\, (r)}_{j} \,, \qquad j 
= 1, \ldots, p-1 \,, \non  \\
\Delta(E^{\pm\, (r)}_{p}) &= \left\{
\begin{array}{ll}
E^{\pm\, (r)}_{p} \otimes e^{2 \eta	H^{(r)}_{p}}  + e^{-2 \eta 
H^{(r)}_{p}} \otimes E^{\pm\, (r)}_{p} &  \mbox{ if  } g^{(r)} = 
C_{p}  \mbox{ i.e., for  } A^{(2)}_{2n}\,, C^{(1)}_{n} \\ \\
E^{\pm\, (r)}_{p} \otimes e^{(\eta +i \pi) H^{(r)}_{p-1}+\eta H^{(r)}_{p}} & \\
\quad
+  e^{ -(\eta + i \pi) H^{(r)}_{p-1}-\eta H^{(r)}_{p}} \otimes 
E^{\pm\, (r)}_{p}
& \mbox{ if  } g^{(r)} = D_{p}  \mbox{ i.e., for  }  A^{(2)}_{2n-1}\,,B^{(1)}_{n}\,, D^{(1)}_{n}
\end{array} \right. \,.
\label{coproductright}
\end{align}

These coproducts satisfy
\be
\left[  \Delta(H^{(r)}_{i}) \,,   \Delta(E^{\pm\, (r)}_{j}) \right] = \pm \alpha_{i}^{(j)} 
\Delta(E^{\pm\, (r)}_{j})\,,  
\label{DeltaHDeltaEr}
\ee
and 
\begin{align}
& \Omega_{ij}^{(r)}\Delta(E_i^{+\, (r)})\Delta(E_j^{-\, (r)})-\Delta(E_j^{-\, 
	(r)})\Delta(E_i^{+\, (r)})\Omega_{ij}^{(r)}\\
&= \begin{cases}
\delta_{i,j}\frac{\sinh\left[2\eta \sum_{k=1}^{p}\alpha_{k}^{(j)}\Delta(H^{(r)}_k)\right]}
{\sinh(2\eta)} & \mbox{for } A_{2n-1}^{(2)}, B_{n}^{(1)}\,, 
D_{n}^{(1)}\\ \\
\delta_{i,j}(1+\delta_{i,p})\frac{\sinh\left[2\eta \sum_{k=1}^{p}\alpha_{k}^{(j)}\Delta(H^{(r)}_k)\right]}
{\sinh(2(1+\delta_{i,p})\eta)} &  \mbox{for } A_{2n}^{(2)}\,, C_{n}^{(1)}
\end{cases} \,,
\label{DeltaEDeltar}
\end{align} 
where $\Omega_{ij}^{(r)}$ is given by 
\be 
\Omega_{ij}^{(r)}=
\begin{cases}
\begin{cases}
e^{i \pi H^{(r)}_{\text{max}(i,j)}}\otimes\id & |i-j|=1 \mbox{ and } 1 \le \text{min}(i,j) \le p-2  \mbox{ and } i\neq p \\
\quad \, \quad \id\otimes\id & \mbox{ otherwise }
\end{cases} & \mbox{for }  A_{2n}^{(2)}\,, C_{n}^{(1)} \,, \\ \\
\begin{cases}
e^{i \pi H^{(r)}_{\text{max}(i,j)}}\otimes\id & |i-j|=1  \mbox{ and } 1 \le \text{min}(i,j) \le p-2, \\
e^{i \pi \left(H_{i}^{(r)}+H_{j}^{(r)}\right)}\otimes \id &  |i-j|=2  
\,\,\mbox{ and } (i=p \text{ or } j=p)\\
\quad \, \quad \id\otimes\id & \mbox{ otherwise }
\end{cases}  & \mbox{for }  A_{2n-1}^{(2)}\,, B_{n}^{(1)}\,, D_{n}^{(1)}
\end{cases}
\,.
\label{omegar}
\ee

\subsection{$\tilde{T}^{\pm}(p)$}

The matrix elements of the asymptotic 
gauge-transformed monodromy matrix 
$\tilde{T}^{\pm}(p)$ (\ref{tildeTpm}) can be expressed in terms 
of the coproducts of the ``left'' and ``right'' generators introduced 
above. 
We now exhibit a set of matrix elements $\tilde{T}^{+}_{ij}(p)$ that includes all $
\Delta_{(N)}(E^{+\, (l)}_{1})\,, \ldots\,, \Delta_{(N)}(E^{+\, 
(l)}_{n-p})$ and all $\Delta_{(N)}(E^{+\, (r)}_{1})\,, \ldots\,, 
\Delta_{(N)}(E^{+\, (r)}_{p})$.

For $j\ne n$ and for all the considered affine algebras, we find that 
\begin{align}
\tilde{T}^{+}_{j+1,j}(p)=\begin{cases}
-\psi\, e^{(\eta+ i 
\pi)\Delta_{(N)}(H_{p-j}^{(r)})+\eta\Delta_{(N)}(H_{p-j+1}^{(r)})}\Delta_{(N)}(E_{p-j}^{+\,(r)})  &  j=1,...,p-1\\
\, 0 & j=p\\
\psi\,e^{(-\eta+i \pi)\Delta_{(N)}(H_{j-p}^{(l)})-\eta 
\Delta_{(N)}(H_{j-p+1}^{(l)})}\Delta_{(N)}(E_{j-p}^{+\,(l)}) & j=p+1,...,n-1
\end{cases} \,,
\label{tildetgens}
\end{align}
where 
\be 
\psi=\frac{e^{-(\kappa\, N-1)\eta}}{2^{N-1}}\sinh(2\eta) \,.
\ee
The set of matrix elements $\{ \Tilde{T}^{+}_{2,1}(p)\,, \ldots \,, 
\Tilde{T}^{+}_{n,n-1}(p)\}$ evidently contains all the generators 
except $\Delta_{(N)}(E_{p}^{+\, (r)})$ and
$\Delta_{(N)}(E_{n-p}^{+\, (l)})$.

For the $p$-th ``right'' generator $ \Delta_{(N)}(E_{p}^{+\, (r)})$ we have  
\begin{align}
\Tilde{T}_{1,\sigma(n)}^{+}(p) =\begin{cases}
0 & p=0\\
-\frac{2}{\sqrt{2}}\psi e^{\eta}\cosh(2\eta)\Delta_{(N)}(E_p^{+\,(r)})
& p=1,...,n \mbox{ for } g^{(r)}=C_{p} \\
& \quad \mbox{ i.e., for }  A^{(2)}_{2n}\,, C^{(1)}_{n} \\
\psi e^{(-\eta+i 
\pi)\Delta_{(N)}(H_{p-1}^{(r)})+\eta\Delta_{(N)}(H_{p}^{(r)})}\Delta_{(N)}(E_p^{+\,(r)}) & p=2,...,n \mbox{ for }g^{(r)}=D_{p} \\
& \quad \mbox{ i.e., for  }  A^{(2)}_{2n-1}\,,B^{(1)}_{n}\,, D^{(1)}_{n}
\end{cases}
\end{align}
where 
\be
\sigma(n)=\begin{cases}
	2n-1 & \mbox{ for } A_{2n-1}^{(2)}\,, D_n^{(1)}\\
	2n & \mbox{ for }B_n^{(1)}\,, C_n^{(1)}\\
	2n+1 & \mbox{ for } A_{2n}^{(2)}
\end{cases}.
\ee 

For the $(n-p)$-th ``left'' generator $\Delta_{(N)}(E_{n-p}^{+\, 
(l)})$ we have (for $p=0, 1, \ldots, n-1$)
\begin{align}
\Tilde{T}_{n+1,\bar{\sigma}(n)}^{+}(p) =\begin{cases}
\psi e^{(-\eta+i 
\pi)\Delta_{(N)}(H_{n-p}^{(l)})}\Delta_{(N)}(E_{n-p}^{+\,(l)}) & \mbox{ for } g^{(l)}=B_{n-p}\\
& \quad \mbox{ 
  i.e., for  } A^{(2)}_{2n}\,, B^{(1)}_{n} \\
-\frac{2}{\sqrt{2}}\psi 
e^{\eta}\cosh(2\eta)\Delta_{(N)}(E_{n-p}^{+\,(l)}) & \mbox{ for } g^{(l)}=C_{n-p}\\
& \quad \mbox{ i.e., for  } A^{(2)}_{2n-1}\,, C^{(1)}_{n} \\
-\psi e^{-\eta\Delta_{(N)}(H_{n-p-1}^{(l)})+(\eta+i 
\pi)\Delta_{(N)}(H_{n-p}^{(l)})}\Delta_{(N)}(E_{n-p}^{+\,(l)}) & 
p\neq n,n-1 \mbox{ for } g^{(l)}=D_{n-p} \\
& \quad \mbox{ 
i.e., for  }  D^{(1)}_{n}
\end{cases} \,,
\end{align}
where
\be
\bar{\sigma}(n)=\begin{cases}
	n & \mbox{ for }  A_{2n}^{(2)}\,, \,A_{2n-1}^{(2)}\,, 
	\,B_{n}^{(1)}\,, C_{n}^{(1)}\\
	n-1 & \mbox{ for }   D_{n}^{(1)}
\end{cases} \,.
\ee

Similar expressions can be found for $\tilde{T}^{-}_{ij}(p)$ in terms 
of $\Delta_{(N)}(E^{-\, (l)}_{1})\,, \ldots\,, \Delta_{(N)}(E^{-\, 
(l)}_{n-p})$ and $\Delta_{(N)}(E^{-\, (r)}_{1})\,, \ldots\,, 
\Delta_{(N)}(E^{-\, (r)}_{p})$.

\section{The Hamiltonian}\label{sec:Hamiltonian}

The transfer matrix (\ref{transfer}) contains \cite{Sklyanin:1988yz} the Hamiltonian 
${\cal H}(p) \sim t'(0,p)$. More explicitly, using the regularity properties
\begin{align}
R(0) &= \xi(0) {\cal P} \,, \non \\
K^{R}(0,p)  &=  \id \,, 
\end{align}
one obtains
\be
{\cal H}(p) = \sum_{k=1}^{N-1} h_{k,k+1} +\frac{1}{2} K^{R\, 
'}_{1}(0,p) + 
\frac{1}{\tr K^{L}(0,p)}\tr_{a} K^{L}_{a}(0,p) h_{N a}\,,
\label{Hamiltonian}
\ee
where the two-site Hamiltonian $h_{k,k+1}$ is given by
\be
h_{k,k+1} = \frac{1}{\xi(0)} {\cal P}_{k, k+1} R'_{k,k+1}(0)  \,.
\label{twosite}
\ee
The Hamiltonian is gauge invariant \cite{Mezincescu:1990uf}
\be
{\cal H}(p) = \sum_{k=1}^{N-1} \tilde{h}_{k,k+1}(p) +\frac{1}{2} \tilde{K}^{R\, 
'}_{1}(0,p) + 
\frac{1}{\tr \tilde{K}^{L}(0,p)}\tr_{a} \tilde{K}^{L}_{a}(0,p) \tilde{h}_{N a}\,,
\label{Hamiltoniangauge}
\ee
where the gauge-transformed two-site Hamiltonian is given by
\begin{align}
\tilde{h}_{k,k+1}(p) &= \frac{1}{\xi(0)} {\cal P}_{k, k+1} 
\tilde{R}'_{k,k+1}(0,p)   \non \\
&= h_{k,k+1} + B'_{k+1}(0,p) - B'_{k}(0,p)\,, 
\label{twositegauge}
\end{align}
where we have used the definition (\ref{gaugeR}) of the 
gauge-transformed R-matrix to pass to the second line.

\subsection{Special cases}

For the special case with $p=0$, the K-matrix $K^{R}(u,0)$ is 
proportional to the identity matrix (\ref{KRpzero}).  It follows\footnote{Indeed, the
second term in (\ref{Hamiltonian}) is evidently proportional to the identity matrix; moreover,
using an identity from \cite{Mezincescu:1990ui, Ahmed:2017mqq, Nepomechie:2017hgw}, one
can show that the third term in (\ref{Hamiltonian}) is also
proportional to the identity matrix.} 
that only the first term in (\ref{Hamiltonian}) contributes \cite{Mezincescu:1990ui}
\be
{\cal H}(0) = \sum_{k=1}^{N-1} h_{k,k+1}  \,.
\ee 

Similarly, for the special case with $p=n$ and $d=2n$,  
the gauge-transformed K-matrix
$\tilde{K}^{R}(u,n)$ is proportional to the identity matrix, see (\ref{mostlyones}). Hence, only the
first term in (\ref{Hamiltoniangauge}) contributes 
\be
{\cal H}(n) = \sum_{k=1}^{N-1} \tilde{h}_{k,k+1}(n) \qquad 
(d=2n) \,.
\ee 
This explains the observation in \cite{Nepomechie:2017hgw} that the 
Hamiltonian for this case is given by a sum 
of two-body terms.  A similar result holds for the special case with 
$p=n$ and $d=2n+1$ \cite{Ahmed:2017mqq}.

\section{Proofs of four lemmas}\label{sec:proofs}

We outline here proofs of Lemmas \ref{lemma:1}, \ref{lemma:5}, 
\ref{lemma:9} and \ref{lemma:13} for any value of $n$. For all of 
these proofs, it is useful to rewrite the  R-matrix \eqref{Rmat} as follows
\begin{equation}
R(u)=c(u)\, R^{(1)}+b(u)\, R^{(2)}+e(u)\, R^{(3)}+\bar{e}(u)\, 
R^{(4)}+R^{(5)}(u)\,,
\label{Rterms}
\end{equation}
where 
\begin{align}
R^{(1)} &=\sum_{\alpha\neq\alpha^{\prime}}e_{\alpha\alpha}\otimes e_{\alpha\alpha}
=\sum_{\alpha}e_{\alpha\alpha}\otimes e_{\alpha\alpha}- 
e_{n+1,n+1}\otimes e_{n+1,n+1}\, (1-\delta_{d,2n}),\label{R1}\\
R^{(2)} &=\sum_{\alpha\neq \beta,\beta^{\prime}}e_{\alpha\alpha}\otimes e_{\beta \beta}=\sum_{\alpha,\beta}e_{\alpha\alpha}\otimes e_{\beta \beta}-\sum_{\beta\neq \beta^{\prime}}e_{\beta\beta}\otimes e_{\beta\beta}-\sum_{\beta}e_{\beta^{\prime}\beta^{\prime}}\otimes e_{\beta\beta},\label{R2}\\
R^{(3)} &=\sum_{\alpha<\beta,\alpha\neq\beta'}e_{\alpha\beta}\otimes e_{\beta\alpha}=\sum_{\alpha<\beta}e_{\alpha\beta}\otimes e_{\beta\alpha}-\sum_{\beta>\frac{d+1}{2}}e_{\beta^{\prime}\beta}\otimes e_{\beta\beta^{\prime}},\label{R3}\\
R^{(4)} &=\sum_{\alpha>\beta,\alpha\neq\beta'}e_{\alpha\beta}\otimes e_{\beta\alpha}=\sum_{\alpha>\beta}e_{\alpha\beta}\otimes e_{\beta\alpha}-\sum_{\beta<\frac{d+1}{2}}e_{\beta^{\prime}\beta}\otimes e_{\beta\beta^{\prime}},\label{R4}\\
R^{(5)}(u) &=\sum_{\alpha,\beta}a_{\alpha\beta}(u)\, e_{\alpha\beta}\otimes e_{\alpha^{\prime}\beta^{\prime}}\label{R5}.
\end{align}
We follow a similar basic strategy for all the proofs: express all the matrices in terms 
of the elementary matrices $e_{ij}$ and the identity matrix $\id$, 
perform the matrix products using the identity
\begin{equation}
e_{ij}\, e_{kl}=\delta_{jk}\, e_{il} \,,
\label{eedeltae}
\end{equation}
and then effectuate the resulting Kronecker deltas. Since many terms 
are generated by this procedure, we use the software {\tt 
Mathematica} to perform the necessary algebra. Since the 
proofs are too long to present all the details, we explain 
the main steps, and point out some of the subtleties. 
We start with the simplest proof (Lemma \ref{lemma:13}), and then work our way to 
the most difficult one (Lemma \ref{lemma:1}).

\subsection{Lemma \ref{lemma:13}}\label{sec:proofz2llemma}

We wish to prove the relation
\be
Z^{(l)}_{1}\, R_{12}(u)\, Z^{(l)}_{1} = Z^{(l)}_{2}\, 
R_{12}(u)\, Z^{(l)}_{2}
\label{zlagain}
\ee
for the $D_{n}^{(1)}$ R-matrix. We begin by rewriting $Z^{(l)}$ (\ref{Z2leftmat}) as
\begin{equation}
Z^{(l)}=\mathbb{I}-e_{n,n}-e_{n+1,n+1}+e_{n,n+1}+e_{n+1,n}\,.
\label{Idminus1}
\end{equation}
The relation (\ref{zlagain}) is in fact separately satisfied by each of the terms in the expression 
\eqref{Rterms} for the R-matrix, which we now discuss in 
turn.

\subsubsection{$R^{(1)}$ and $R^{(2)}$}

Since we consider here only the $D_{n}^{(1)}$ R-matrix, here 
$d=2n$; therefore, the second term in \eqref{R1} is absent.
For $R^{(1)}$ and $R^{(2)}$, the sums in $\alpha$ and $\beta$ do not 
have any restriction of the type $\alpha<\beta$ or $\alpha>\beta$; 
hence, it is straightforward to show using \eqref{eedeltae} that
\begin{align}
Z_1^{(l)}\, R^{(1)}\, Z_1^{(l)} &=Z_2^{(l)}\, R^{(1)}\, Z_2^{(l)} \,,\non \\
Z_1^{(l)}\, R^{(2)}\, Z_1^{(l)} &=Z_2^{(l)}\, R^{(2)}\, Z_2^{(l)} \,.
\end{align}

\subsubsection{$R^{(3)}$ and $R^{(4)}$}

These terms require much more effort.  Let us start by considering
the first term in $R^{(3)}$, and calculating 
\begin{equation}
Z_1^{(l)}\left(\sum_{\alpha<\beta}e_{\alpha\beta}\otimes e_{\beta\alpha}\right)Z_1^{(l)}.
\end{equation}
Using the relation \eqref{eedeltae} we obtain an expression depending
on Kronecker deltas.  But we cannot directly effectuate these Kronecker
deltas to evaluate the sums because of the condition $ \alpha<\beta$.
We can put terms such as $ \delta_{n,\beta}\,\delta_{n+1,\alpha}$,
$\delta_{n,\beta}\, \delta_{n,\alpha} $ and $\delta_{n+1,\beta}\,
\delta_{n+1,\alpha} $ to zero, because they do not obey $ \alpha<\beta
$.  After doing this, we remain with expressions such as 
\begin{equation}
\sum_{\alpha<\beta}e_{n,\beta}\otimes e_{\beta,\alpha}\, \delta_{n,\alpha}.
\label{Zlcondition1a}
\end{equation}
Notice that we cannot simply set $\alpha=n$ in this expression.  In order to satisfy the
condition $\alpha<\beta$, if $\alpha=n$, then $\beta \in 
\{n+1,\,...\,,2n\}$.  Hence, we can rewrite
\eqref{Zlcondition1a} as
\begin{equation}
\sum_{\alpha<\beta}e_{n,\beta}\otimes e_{\beta,\alpha}\, \delta_{n,\alpha}=e_{n,n+1}\otimes e_{n+1,n}+\sum_{\beta=n+2}^{2n}e_{n,\beta}\otimes e_{\beta, n},
\label{Zlcondition1b}
\end{equation}
where we separate the term with $\beta=n+1 $ from the sum, since this helps to cancel with other terms.
For the same reason, we can rewrite
\begin{equation}
\sum_{\alpha<\beta}e_{n,\beta}\otimes e_{\beta, \alpha}\, \delta_{n+1,\alpha}=\sum_{\beta=n+2}^{2n}e_{n,\beta}\otimes e_{\beta, n+1}.
\label{Zlcondition2}
\end{equation}
Using similar logic with all of the terms, we obtain
\begin{equation}
Z_1^{(l)}\left(\sum_{\alpha<\beta}e_{\alpha\beta}\otimes e_{\beta\alpha}\right)Z_1^{(l)}-Z_2^{(l)}\left(\sum_{\alpha<\beta}e_{\alpha\beta}\otimes e_{\beta\alpha}\right)Z_2^{(l)}=e_{1+n,n}\otimes e_{1+n,n}-e_{n,1+n}\otimes e_{n,1+n}.
\label{Zlt3part1}
\end{equation}

We still must consider the contribution of the second term in $R^{(3)}$ 
\begin{equation}
Z_1^{(l)}\left(-\sum_{\beta>\frac{d+1}{2}}e_{\beta^{\prime}\beta}\otimes e_{\beta\beta^{\prime}}\right)Z_1^{(l)}.
\label{Zlcondition3a}
\end{equation}
Notice that, since $d=2n$, the condition $\beta>\frac{d+1}{2}$ is
equivalent to $\beta\ge n+1$.  Due to this condition, all terms
with $\delta_{n,\beta}$ and $\delta_{n+1,2n+1-\beta}$ must vanish.  
Taking this into account, we obtain
\begin{equation}
Z_1^{(l)}\left(-\sum_{\beta>\frac{d+1}{2}}e_{\beta^{\prime}\beta}\otimes e_{\beta\beta^{\prime}}\right)Z_1^{(l)}
-Z_2^{(l)}\left(-\sum_{\beta>\frac{d+1}{2}}e_{\beta^{\prime}\beta}\otimes e_{\beta\beta^{\prime}}\right)Z_2^{(l)}=-e_{1+n,n}\otimes e_{1+n,n}+e_{n,1+n}\otimes e_{n,1+n},
\label{Zlt3part2}
\end{equation}
which exactly cancels with \eqref{Zlt3part1}. 
We conclude that $R^{(3)}$ satisfies
\be
Z_1^{(l)}\, R^{(3)}\, Z_1^{(l)} = Z_2^{(l)}\, R^{(3)}\, Z_2^{(l)}\,.
\ee

We prove that $R^{(4)}$ satisfies
\be
Z_1^{(l)}\, R^{(4)}\, Z_1^{(l)} = Z_2^{(l)}\, R^{(4)}\, Z_2^{(l)}
\ee
using the same arguments presented for $R^{(3)}$, but considering $\alpha>\beta$ instead of $\alpha<\beta$.

\subsubsection{$R^{(5)}(u)$}

For $R^{(5)}(u)$, there are no restrictions on the sums over $\alpha$ and $\beta$; 
hence, we can directly effectuate all the Kronecker deltas. However, 
doing this is not enough to show that 
\begin{equation}
Z_1^{(l)}\left(\sum_{\alpha,\beta}a_{\alpha\beta}(u)\, e_{\alpha\beta}\otimes e_{\alpha^{\prime}\beta^{\prime}}\right)Z_1^{(l)}
-Z_2^{(l)}\left(\sum_{\alpha,\beta}a_{\alpha\beta}(u)\, e_{\alpha\beta}\otimes e_{\alpha^{\prime}\beta^{\prime}}\right)Z_2^{(l)} = 0 \,.
\end{equation}
To this end, it is useful to separate all the terms with 
$\alpha,\beta \in \{n,n+1\}$ from the sums. For example, 

\begin{align}
\sum_{\beta}a_{n,\beta}(u)\, e_{n+1,\beta}\otimes 
e_{n,\beta^{\prime}}=& a_{n,n}(u)\, e_{n+1,n}\otimes 
e_{n,n+1}+a_{n,n+1}(u)\, e_{n+1,n+1}\otimes e_{n,n}+\nonumber\\
& +\sum_{\beta=1}^{n-1}a_{n,\beta}(u)\, e_{n+1,\beta}\otimes 
e_{n,\beta^{\prime}}+\sum_{\beta=n+2}^{2n}a_{n,\beta}(u)\, e_{n+1,\beta}\otimes e_{n,\beta^{\prime}} \,.
\end{align}
By doing this, we find that all the terms without sums 
cancel. The remaining 
terms can also be seen to cancel by using the following properties of the functions $a_{\alpha\beta}(u)$ (\ref{Ra})
for $D_n^{(1)}$ 
\begin{align}
a_{n,n} &=a_{n+1,n+1},\nonumber\\
a_{n,n+1} &=a_{n+1,n},\nonumber\\
a_{n,\beta} &=a_{n+1,\beta} \quad \text{for}\quad 1\le \beta\le n-1\quad \text{and for }\quad n+2\le\beta\le 2n,\nonumber \\ 
a_{\beta,n} &=a_{\beta,n+1} \quad \text{for}\quad 1\le \beta\le n-1\quad \text{and for }\quad n+2\le\beta\le 2n.
\end{align}
We conclude that
\be
Z_1^{(l)}\, R^{(5)}(u)\, Z_1^{(l)} = Z_2^{(l)}\, R^{(5)}(u)\, 
Z_2^{(l)}\,,
\ee
which concludes the proof of (\ref{zlagain}).

\subsection{Lemma \ref{lemma:9}}\label{sec:proofz2rlemma}

We now turn to the proof of the relations
\begin{align}
Z^{(r)}_{1}\, R_{12}(u)\, Z^{(r)}_{1} &= Y^{t}_{2}(u)\, 
R_{12}(u)\, Y^{t}_{2}(u)\,,  \non \\
Z^{(r)}_{2}\, R_{12}(u)\, Z^{(r)}_{2} &= Y_{1}(u)\, 
R_{12}(u)\, Y_{1}(u)\,,
\label{Z2rightR2b}
\end{align}
for the $A_{2n-1}^{(2)}$, $B_{n}^{(1)}$ and $D_{n}^{(1)}$ R-matrices.
We begin by rewriting $Z^{(r)}$ and $Y(u)$ (\ref{Z2rightmat}) as follows
\begin{align}
& Z^{(r)}=\mathbb{I}-e_{1,1}-e_{d,d}+e_{1,d}+e_{d,1} \,,\non \\
& Y(u)=\mathbb{I}-e_{1,1}-e_{d,d}+e^{-u}e_{1,d}+e^u e_{d,1} \,.
\label{Idminus2}
\end{align}

The rest of the proof is very similar to the one for Lemma \ref{lemma:13}
(\ref{zlagain}). However, whereas in the previous case
all the terms are written in such a way that $ \alpha,\beta \in 
\{n,n+1\}$ appear explicitly and not inside the sums, here we should write all
the terms in such a way that $\alpha,\beta \in \{1,d\}$ appear explicitly.
Another difference is that now not all the terms in the expression 
(\ref{Rterms}) for the R-matrix separately satisfy the relations 
(\ref{Z2rightR2b}). Indeed, the linear combination $R^{(3)}
+e^{u}\, R^{(4)}$ satisfies these relations, but not $R^{(3)}$ and $R^{(4)}$ separately. 
Otherwise, all the intermediate strategies are analogous.  At the 
end, we must use the following properties of the functions $a_{\alpha\beta}(u)$ (\ref{Ra})
for $A_{2n-1}^{(2)}$, $B_n^{(1)}$ and $D_n^{(1)}$
\begin{align}
a_{1,1} &=a_{d,d},\non \\
a_{d,1} &=e^{2u}\, a_{1,d},\non\\
a_{d,\beta} &=e^{u}\, a_{1,\beta} \quad \text{for} \quad 2\le 
\beta\le d-1\,,\non\\
a_{\beta,d} &=e^{-u}\, a_{\beta,1}\quad \text{for} \quad 2\le 
\beta\le d-1\,.
\end{align}

\subsection{Lemma \ref{lemma:5}}\label{sec:proofdualitylemma}

We now present some details about our proof of the duality relation
\begin{align}
U_2\, R_{12}(u)\, U_2 &= W_1(u)\, R_{12}(u)\, W_1(u)  \label{duality1a}
\end{align}
for the $C_n^{(1)}$ and $D_n^{(1)}$ R-matrices, for which $d=2n$. 
In contrast with the previous proofs (\ref{Idminus1}),  
(\ref{Idminus2}), the matrices $U$ and $W(u)$ \eqref{UWmats} cannot 
be expressed in the form $(\mathbb{I}-\text{few terms} )$. We 
rewrite these matrices instead as
\begin{equation}
U=\sum_{i=1}^{n}\left(e_{i,n+i}+e_{n+i,i}\right),
\label{newU}
\end{equation}
\noindent
and 
\begin{equation}
W(u)=\sum_{i=1}^{n}\left(e^{-\frac{u}{2}}\, 
e_{i,n+i}+e^{\frac{u}{2}}\, e_{n+i,i}\right).
\label{newW}
\end{equation}
We now proceed to analyze separately the contributions of the terms in 
the expression (\ref{Rterms}) for the R-matrix to the relation
(\ref{duality1a}).

\subsubsection{$R^{(1)}$ and $R^{(2)}$}
 
After applying the rule \eqref{eedeltae}, we must deal with the
ranges of the sums.  The ranges for the sums in \eqref{newU} and
\eqref{newW}  (from 1 to $n$) are different from the ones in \eqref{R1} 
and \eqref{R2}  (from 1 to $2n$).  We cannot effectuate the Kronecker deltas
to evaluate the sums in $R^{(1)}$ unless we split those sums into two
ranges: $1\le\alpha\le n$ and $n+1\le\alpha\le 2n$. In \eqref{duality1a}
we write the $U_2$ on the left hand side of $R^{(1)}$
with a sum in $i$, and the $U_{2}$ on the right hand side with a sum in $j$.
For the range $1\le\alpha\le n$, all the terms with $\delta_{i+n,\alpha}$ and $\delta_{j+n,\alpha}$ are zero, 
because $\alpha $ is always smaller than $n+i$.  For $n+1\le\alpha\le 2n$, 
all the terms with $\delta_{i,\alpha}$ and $\delta_{j,\alpha}$ are zero,
because max$(i)$ and max$(j)$ are $n$, while $\alpha$
is always greater or equal to $n+1$. After applying such arguments, we obtain
\begin{align}
U_2\, R^{(1)}\, U_2=\sum_{\alpha=1}^{n}e_{\alpha,\alpha}\otimes e_{\alpha+n,\alpha+n}
+\sum_{\alpha=n+1}^{2n}e_{\alpha,\alpha}\otimes e_{\alpha-n,\alpha-n} 
\,.
\label{URUt1part1}
\end{align}
By applying analogous arguments for the terms with $W_1(u)$, we find
\begin{align}
W_1(u)\, R^{(1)}\, 
W_1(u)=\sum_{\alpha=1}^{n}e_{\alpha+n,\alpha+n}\otimes 
e_{\alpha,\alpha} +\sum_{\alpha=n+1}^{2n}e_{\alpha-n,\alpha-n}\otimes e_{\alpha,\alpha}\,.
\label{URUt1part2}
\end{align}
We conclude that
\be
U_2\, R^{(1)}\, U_2 = W_1(u)\, R^{(1)}\, W_1(u)\,,
\ee
since the right-hand-sides of \eqref{URUt1part1} and 
\eqref{URUt1part2} become identical upon redefining the $\alpha$'s in 
the sums. We prove in a similar way that $R^{(2)}$ satisfies
\be
U_2\, R^{(2)}\, U_2 = W_1(u)\, R^{(2)}\, W_1(u)\,.
\ee

\subsubsection{$R^{(3)}$ and $R^{(4)}$}

The duality relation is not satisfied separately by $R^{(3)}$ and 
$R^{(4)}$, but is instead satisfied by the linear combination 
$R^{(3)}+e^u\, R^{(4)}$. That is,
\begin{equation}
U_2\left( R^{(3)}+e^u\, R^{(4)}\right) U_2=W_1(u) \left( 
R^{(3)}+e^u\, R^{(4)}\right) W_1(u) \,.
\label{URUt3t4}
\end{equation}
In order to manage the cases with $\alpha<\beta$ and $\alpha>\beta$,
we split the sums over $\alpha$ and $\beta$ into four ranges:
\begin{align}
& 1\le\alpha\le n \quad\text{and}\quad 1\le\beta\le n \,, \non\\
& 1\le\alpha\le n \quad\text{and}\quad n+1\le\beta\le 2n \,, \non\\ 
& n+1\le\alpha\le 2n \quad\text{and}\quad  1\le\beta\le n\,, \non\\  
& n+1\le\alpha\le 2n \quad\text{and}\quad n+1\le\beta\le 2n \,. 
\label{cases}
\end{align}
For each of these ranges, we put to zero terms that contain Kronecker deltas
where $\alpha$ and $\beta$ are outside of the relevant interval.
Again, at the end, it is necessary to redefine $ \alpha $ and $ \beta
$ on the sums to see that \eqref{URUt3t4} is satisfied.

\subsubsection{$R^{(5)}(u)$}

For this term we also split the sums over $\alpha$ and $\beta$ into the 
four ranges (\ref{cases}). All the other strategies are similar to 
the ones presented above, and we obtain
\be
U_2\, R^{(5)}(u)\, U_2 = W_1(u)\, R^{(5)}(u)\, W_1(u)\,.
\ee

\subsection{Lemma \ref{lemma:1} for $d=2n$}\label{sec:prooflemma1a}

In order to prove 
\be
\left[\tilde{R}_{12}^{+}(p)\,, \tilde{K}^{R}_{2}(u,p) \right] = 0 
\label{lemma1again}
\ee
for any value of $n$, we proceed in
three steps: finding an explicit expression for the 
gauge-transformed R-matrix $\tilde{R}_{12}(u,p)$, 
performing the limit $u\rightarrow\infty$ in $e^{-u}\, 
\tilde{R}_{12}(u,p)$ to obtain $\tilde{R}_{12}^{+}(p)$, and
finally evaluating the commutator. We consider here the case $d=2n$, 
leaving the case $d=2n+1$ for the following subsection.

\subsubsection{Finding $\tilde{R}_{12}(u,p)$}

In order to obtain an explicit expression for the gauge-transformed R-matrix
$\tilde{R}_{12}(u,p)$ \eqref{gaugeR}, it is useful to rewrite $B(u)$ 
\eqref{Bgauge} in terms of elementary matrices 
\begin{equation}
B(u)=e^{\frac{u}{2}}\sum_{i=1}^{p}e_{i,i}+\sum_{i=p+1}^{n}e_{i,i}
+\sum_{i=n+1}^{2n-p}e_{i,i}+e^{-\frac{u}{2}}\sum_{i=2n-p+1}^{2n}e_{i,i} \,,
\end{equation}
\noindent
for $1\le p\le n-1$.\footnote{For $p=0$ and $p=n$, 
$\tilde{K}^{R}(u,p)\propto \mathbb{I}$, so \eqref{lemma1again} is trivially satisfied.}

We now point out some useful simplifications for the contributions 
from each of the terms in the expression (\ref{Rterms}) for the R-matrix.

Since $B(u)$ is a diagonal matrix,
\begin{align}
B_1(u)\, R^{(1)}\, B_1(-u) &= R^{(1)}\,, \\
B_1(u)\, R^{(2)}\, B_1(-u) &= R^{(2)}\,.
\end{align}

Let us now consider the first term in $B_1(u)\, R^{(3)}\, B_1(-u)$, where 
$\alpha < \beta$. After applying the rule 
\eqref{eedeltae}, we obtain terms such as 
\begin{equation}
\sum_{\alpha < 
\beta}\sum_{i=2n-p+1}^{2n}\sum_{j=n+1}^{2n-p}e_{i,j}\otimes 
e_{\beta,\alpha}\, \delta_{i,\alpha}\, \delta_{j,\beta} \,,
\end{equation}
\noindent
for example. Several terms like this appear, but they are all equal 
to zero, because the $\delta $'s force $\alpha=i$ and $\beta=j$; but 
$i>j$ in this sum, which contradicts the condition $\alpha<\beta$. 
For the second term in $B_1(u)\, R^{(3)}\, B_1(-u)$, several terms 
are zero because the Kronecker deltas force $\beta$ to have values 
that are not greater than $\frac{d+1}{2}$. Similar arguments can be 
used for $B_1(u)\, R^{(4)}\, B_1(-u)$.

For $B_1(u)\, R^{(5)}(u)\, B_1(-u)$, after applying the rule
\eqref{eedeltae}, we can directly use the $\delta$'s to evaluate the
sums, because there are no restrictions on the $\alpha$'s and
$\beta$'s.  
The functions $a_{i,j}(u) $ have different expressions depending on 
whether $i=j$, $i<j$ or $i>j$. For later convenience,
we separately calculate the contributions from each
of these three cases.  For example, consider the term
\begin{equation}
\sum_{i=2n-p+1}^{2n}\sum_{j=1}^{p}a_{i,j}(u)\, e_{i,j}\otimes 
e_{i^{\prime},j^{\prime}} \,.
\end{equation}
This term contributes only to $i>j$, due to the ranges in the sums 
and the fact $2n-p+1>p$.

We refrain from displaying the final result for 
$\tilde{R}_{12}(u,p)$, which is quite lengthy
even after the simplifications noted above.

\subsubsection{Performing the large-$u$ limit}

We now proceed to perform the limit $u\rightarrow\infty$ in $e^{-u}\, 
\tilde{R}_{12}(u,p)$. To this end, we need the following results
{\allowdisplaybreaks 
\begin{align}
&\lim_{u\rightarrow \infty}e^{-u}e(u)=0=\lim_{u\rightarrow \infty}e^{-\frac{u}{2}}e(u)=\lim_{u\rightarrow \infty}e^{-\frac{3u}{2}}\bar{e}(u)=\lim_{u\rightarrow \infty}e^{-2u}\bar{e}(u),\non\\
&\lim_{u\rightarrow \infty}e^{-u}a_{\alpha\beta}^{(3)}(u)=0=\lim_{u\rightarrow \infty}e^{-\frac{u}{2}}a^{(3)}_{\alpha\beta}(u)=\lim_{u\rightarrow \infty}e^{-2u}a_{\alpha\beta}^{(4)}(u)=\lim_{u\rightarrow \infty}e^{-\frac{3u}{2}}a^{(4)}_{\alpha\beta}(u),\non\\
&\mathfrak{b}\equiv\lim_{u\rightarrow \infty}e^{-u}b(u)=\frac{1}{2}e^{-\kappa \eta},\non\\
&\mathfrak{c}\equiv\lim_{u\rightarrow \infty}e^{-u}c(u)=\frac{1}{2}e^{-(\kappa+2)\eta},\non\\
&\mathfrak{e}\equiv\lim_{u\rightarrow \infty}e(u)=-e^{-\kappa \eta}\sinh(2\eta)=\lim_{u\rightarrow \infty}e^{-u}\bar{e}(u),\non\\
&\mathfrak{a}^{(1)}\equiv\lim_{u\rightarrow \infty}e^{-u}a_{\alpha\beta}^{(1)}(u)=\frac{1}{2}e^{-(\kappa-2)\eta},\non\\
&\mathfrak{a}^{(2)}\equiv\lim_{u\rightarrow \infty}e^{-u}a_{\alpha\beta}^{(2)}(u)=\frac{1}{2}e^{-\kappa\eta}\non\\
&\mathfrak{a}^{(3)}_{\alpha,\beta}\equiv\lim_{u\rightarrow \infty}a_{\alpha\beta}^{(3)}(u)=e^{-\kappa\eta}\sinh(2\eta)\left(\delta_1^2e^{2(\kappa+\bar{\alpha}-\bar{\beta})\eta}\epsilon_\alpha\epsilon_\beta-\delta_{\alpha,\beta^{\prime}}\right),\non\\
&\mathfrak{a}^{(4)}_{\alpha,\beta}\equiv\lim_{u\rightarrow 
\infty}e^{-u}a_{\alpha\beta}^{(4)}=e^{-\kappa\eta}\sinh(2\eta)\left(e^{2(\bar{\alpha}-\bar{\beta})\eta}\epsilon_\alpha\epsilon_\beta-\delta_{\alpha,\beta^{\prime}}\right)\,,
\end{align}}
where 
\begin{equation}
a_{\alpha,\beta}(u)=\begin{cases}
 a^{(1)}_{\alpha,\beta}(u)\text{ for } \alpha=\beta, \,\alpha\neq \alpha^{\prime}\\
 a^{(2)}_{\alpha,\beta}(u)\text{ for } \alpha=\beta, \,\alpha= \alpha^{\prime}\\
 a^{(3)}_{\alpha,\beta}(u)\text{ for } \alpha<\beta\\
 a^{(4)}_{\alpha,\beta}(u)\text{ for } \alpha>\beta
\end{cases} \,,
\label{a's}
\end{equation}
\noindent
and the definition of $a_{\alpha,\beta}^{(i)}(u)$ can be read off directly from (\ref{Ra}).

With the help of these results, we find that $\tilde{R}_{12}^{+}(p)$ 
(\ref{Rtildeasym}) is given, for $d=2n$ and $1 \le p \le n-1$, by 
{\allowdisplaybreaks
\begin{align}
\tilde{R}_{12}^{+}(p) = &\, \mathfrak{c}\sum_{\alpha}e_{\alpha,\alpha}\otimes e_{\alpha,\alpha}+\mathfrak{b}\sum_{\alpha \neq\beta,\beta^{\prime}}e_{\alpha,\alpha}\otimes e_{\beta,\beta}-\mathfrak{e}\left(\sum_{\beta=p+1}^{n}+\sum_{\beta=2n-p+1}^{2n}\right)e_{\beta^\prime,\beta}\otimes e_{\beta,\beta^{\prime}}\non\\
& +\mathfrak{e}\left(\sum_{\substack{\alpha,\beta=1\\ \alpha>\beta}}^{p}+\sum_{\substack{\alpha,\beta=p+1\\ \alpha>\beta}}^{n}+\sum_{\substack{\alpha,\beta=n+1\\ \alpha>\beta}}^{2n-p}+\sum_{\substack{\alpha,\beta=2n-p+1\\ \alpha>\beta}}^{2n}+\sum_{\alpha=1}^{p}\sum_{\beta=2n-p+1}^{2n}+\sum_{\alpha=n+1}^{2n-p}\sum_{\beta=p+1}^{n}\right)e_{\alpha,\beta}\otimes e_{\beta,\alpha}\non\\
&+\mathfrak{a}^{(1)}\sum_{\alpha}e_{\alpha,\alpha}\otimes e_{\alpha^{\prime},\alpha^{\prime}}+\sum_{\alpha=1}^{p}\sum_{\beta=2n+1-p}^{2n}\mathfrak{a}^{(3)}_{\alpha,\beta}\,e_{\alpha,\beta}\otimes e_{\alpha^{\prime},\beta^{\prime}}\non\\
& +\left(\sum_{\substack{\alpha,\beta=1\\ \alpha>\beta}}^{p}+\sum_{\substack{\alpha,\beta=p+1\\ \alpha>\beta}}^{n}+\sum_{\substack{\alpha,\beta=n+1\\ \alpha>\beta}}^{2n-p}+\sum_{\substack{\alpha,\beta=2n-p+1\\ \alpha>\beta}}^{2n}+\sum_{\alpha=n+1}^{2n-p}\sum_{\beta=p+1}^{n}\right)\mathfrak{a}^{(4)}_{\alpha,\beta}\,e_{\alpha,\beta}\otimes e_{\alpha^{\prime},\beta^{\prime}}.
\label{Rtildeplus}
\end{align}}

\subsubsection{Evaluating the commutator}

In order to evaluate the commutator (\ref{lemma1again}), we rewrite 
$\tilde{K}^{R}(u,p)$ \eqref{mostlyones} in terms of elementary 
matrices, and obtain 
\begin{equation}
 \tilde{K}^{R}_{2}(u,p)= \id \otimes \left[\sum_{i=1}^{p}+\left(\frac{\gamma e^u+1}{\gamma+e^u}\right)\sum_{i=p+1}^{n}
 +\left(\frac{\gamma 
 e^u+1}{\gamma+e^u}\right)\sum_{i=n+1}^{2n-p}+\sum_{i=2n-p+1}^{2n}\right]e_{i,i}  \,.
 \label{Ktilde}
\end{equation}
It is then just a matter of applying the same ideas presented above, 
and putting to zero all the terms that do not belong to the relevant range.  
In this way, one can see that each of the terms in
\eqref{Rtildeplus} commutes with \eqref{Ktilde}.

\subsection{Lemma \ref{lemma:1} for $d=2n+1$}\label{sec:prooflemma1b}

The cases where $d=2n+1$ can be analogously proved. However, it is 
more suitable to separate the ``middle'' terms in 
$B(u)$ and $\tilde{K}^{R}_{2}(u,p)$, i.e. we set
\begin{equation}
B(u)=e^{\frac{u}{2}}\sum_{i=1}^{p}e_{i,i}+\sum_{i=p+1}^{n}e_{i,i}+e_{n+1,n+1}+\sum_{i=n+2}^{2n-p+1}e_{i,i}+e^{-\frac{u}{2}}\sum_{i=2n-p+2}^{2n+1}e_{i,i}
\end{equation}
\noindent
and
\begin{equation}
\tilde{K}^{R}_2(u,p)= \id \otimes e_{n+1,n+1} + \id \otimes 
\left[\sum_{i=1}^{p} + \left(\frac{\gamma 
e^u+1}{\gamma+e^u}\right)\sum_{i=p+1}^{n}+\left(\frac{\gamma 
e^u+1}{\gamma+e^u}\right)\sum_{i=n+2}^{2n-p+1}+\sum_{i=2n-p+2}^{2n+1}\right]e_{i,i}\,,
\label{Ktilde2}
\end{equation}
\noindent
for $1\le p\le n-1$. For this case, $\tilde{R}_{12}^{+}(p)$ is given by
{\allowdisplaybreaks
\begin{align}
\tilde{R}_{12}^{+}(p) = &\, \mathfrak{c}\sum_{\alpha\neq \alpha^{\prime}}e_{\alpha,\alpha}\otimes e_{\alpha,\alpha}+\mathfrak{b}\sum_{\alpha \neq\beta,\beta^{\prime}}e_{\alpha,\alpha}\otimes e_{\beta,\beta}-\mathfrak{e}\left(\sum_{\beta=2n-p+2}^{2n+1}+\sum_{\beta=p+1}^{n}\right)e_{\beta^\prime,\beta}\otimes e_{\beta,\beta^{\prime}}\non\\
&+\mathfrak{e}\left(\sum_{\alpha=1}^{p}\sum_{\beta=2n-p+2}^{2n+1}+\sum_{\substack{\alpha,\beta=1\\ \alpha>\beta}}^{p}+\sum_{\substack{\alpha,\beta=p+1\\ \alpha>\beta}}^{n}+\sum_{\substack{\alpha,\beta=n+2\\ \alpha>\beta}}^{2n-p+1}+\sum_{\substack{\alpha,\beta=2n-p+2\\ \alpha>\beta}}^{2n+1}+\sum_{\alpha=n+2}^{2n-p+1}\sum_{\beta=p+1}^{n}\right)e_{\alpha,\beta}\otimes e_{\beta,\alpha}\non\\
&+\mathfrak{a}^{(1)}\left(\sum_{\alpha=1}^{n}+\sum_{\alpha=n+2}^{2n+1}\right)e_{\alpha,\alpha}\otimes e_{\alpha^{\prime},\alpha^{\prime}}+\sum_{\alpha=1}^{p}\sum_{\beta=2n+2-p}^{2n+1}\mathfrak{a}^{(3)}_{\alpha,\beta}\,e_{\alpha,\beta}\otimes e_{\alpha^{\prime},\beta^{\prime}}\non\\
&+\left(\sum_{\substack{\alpha,\beta=1\\ \alpha>\beta}}^{p}+\sum_{\substack{\alpha,\beta=p+1\\ \alpha>\beta}}^{n}+\sum_{\substack{\alpha,\beta=n+2\\ \alpha>\beta}}^{2n-p+1}+\sum_{\substack{\alpha,\beta=2n-p+2\\ \alpha>\beta}}^{2n+1}+\sum_{\alpha=n+2}^{2n-p+1}\sum_{\beta=p+1}^{n}\right)\mathfrak{a}^{(4)}_{\alpha,\beta}\,e_{\alpha,\beta}\otimes e_{\alpha^{\prime},\beta^{\prime}}\non\\
& +\mathfrak{a}^{(2)}e_{n+1,n+1}\otimes e_{n+1,n+1}-\mathfrak{e}\,e^{2(n+1)\eta}\sum_{\beta=p+1}^{n}e^{-2\bar{\beta}\eta}e_{n+1,\beta}\otimes e_{n+1,\beta^{\prime}}\non\\
& -\mathfrak{e}\,e^{-2(n+1)\eta}\sum_{\beta=n+2}^{2n-p+1}e^{2\bar{\beta}\eta}e_{\beta,n+1}\otimes e_{\beta^{\prime},n+1}\non\\
& +\mathfrak{e}\left(\sum_{\alpha=p+1}^{n}e_{n+1,\alpha}\otimes e_{\alpha,n+1}+\sum_{\alpha=n+2}^{2n-p+1}e_{\alpha,n+1}\otimes e_{n+1,\alpha}\right).
\label{Rtildeplus2}
\end{align}}

For $p=n$, it is suitable to write $B(u)$ and $\tilde{K}^{R}_2(u,p)$ as 
\begin{equation}
B(u)=e^{\frac{u}{2}}\sum_{i=1}^{n}e_{i,i}+e_{n+1,n+1}+\sum_{i=n+2}^{2n+1}e_{i,i}\,,
\end{equation}
\noindent
and
\begin{equation}
\tilde{K}^{R}_2(u,p)= \id\otimes\id-\id \otimes e_{n+1,n+1} +\left(\frac{\gamma 
	e^u+1}{\gamma+e^u}\right)\id \otimes e_{n+1,n+1} \,.
\label{Ktilde2b}
\end{equation}
For this case, $\tilde{R}_{12}^{+}(p)$ is given by
{\allowdisplaybreaks
\begin{align}
\tilde{R}_{12}^{+}(p) = &\, \mathfrak{c}\sum_{\alpha\neq \alpha^{\prime}}e_{\alpha,\alpha}\otimes e_{\alpha,\alpha}+\mathfrak{b}\sum_{\alpha \neq\beta,\beta^{\prime}}e_{\alpha,\alpha}\otimes e_{\beta,\beta}-\mathfrak{e}\sum_{\beta=n+2}^{2n+1}e_{\beta^{\prime},\beta}\otimes e_{\beta,\beta^{\prime}}\non\\
&+ \mathfrak{e}\left(\sum_{\alpha=1}^{n}\sum_{\beta=n+2}^{2n+1}+\sum_{\substack{\alpha,\beta=1\\ \alpha>\beta}}^{n}+\sum_{\substack{\alpha,\beta=n+2\\ \alpha>\beta}}^{2n+1}\right)e_{\alpha,\beta}\otimes e_{\beta,\alpha}\non\\
& +\mathfrak{a}^{(1)}\left(\sum_{\alpha=1}^{n}+\sum_{\alpha=n+2}^{2n+1}\right)e_{\alpha,\alpha}\otimes e_{\alpha^{\prime},\alpha^{\prime}}+\mathfrak{a}^{(2)}\,e_{n+1,n+1}\otimes e_{n+1,n+1}\non\\
&+\sum_{\alpha=1}^{n}\sum_{\beta=n+2}^{2n+1}\mathfrak{a}^{(3)}_{\alpha,\beta}\,e_{\alpha,\beta}\otimes e_{\alpha^{\prime},\beta^{\prime}}+ \left(\sum_{\substack{\alpha,\beta=1\\ \alpha>\beta}}^{n}+\sum_{\substack{\alpha,\beta=n+2\\ \alpha>\beta}}^{2n+1}\right)\mathfrak{a}^{(4)}_{\alpha,\beta}\,e_{\alpha,\beta}\otimes e_{\alpha^{\prime},\beta^{\prime}}\,.
\label{Rtildeplus3}
\end{align}}

For $p=0$, $\tilde{K}^{R}(u,p) \propto \mathbb{I}$, so \eqref{lemma1again} is trivially satisfied.
 

\providecommand{\href}[2]{#2}\begingroup\raggedright\endgroup

\end{document}